\documentclass[10pt]{amsart}
\usepackage[margin=1.375in]{geometry}
\usepackage{amsthm, amsmath,amsfonts,amssymb,euscript,hyperref,graphics,color,slashed}
\usepackage{graphicx}
\usepackage{comment}
\usepackage{import}
\usepackage{tikz}
\usepackage{latexsym}

\def\a{\mathfrak{a}}

\def\B{{\mathcal B}}
\def\Bb{\overline{\mathcal B}}
\def\C{{\mathcal C}}
\def\Cb{\overline{C}}
\def\E{{\mathcal E}}
\def\F{{\mathcal F}}

\def\H{{\mathcal H}}
\def\Hb{{\overline{\mathcal{H}}}}
\def\I{{\mathcal I}}
\def\lb{\underline{l}}
\def\Lb{\underline{L}}
\def\M{{\mathcal M}}
\def\L{{\mathcal L}}
\def\P{\mathcal{P}}
\def\Pb{\overline{P}}
\def\Ph{\widehat{P}}
\def\Phb{\bar{\Ph}}
\def\Q{{\mathcal Q}}
\def\R{\mathcal R}

\def\T{\mathcal T}
\def\ub{\underline{u}}
\def\X{\mathcal X}
\def\Xib{\overline{\Xi}}

\def\deltab{\bar{\delta}}

\def\chib{\underline{\chi}}
\def\chibh{\widehat{\underline{\chi}}}
\def\chih{\widehat{\chi}}

\def\Lb{\underline{L}}

\def\M{\mathcal{M}}
\def\I{\mathcal{I}}
\def\H{\mathcal{H}}
\def\O{\mathcal{O}}
\def\H{\mathcal{H}}
\def\S{\mathcal{S}}
\def\tr{\text{tr}}

\def\tensor{\tilde{\otimes}}
\def\ub{\underline{u}}

\def\S{\mathcal{S}}
\def\L{\mathcal{L}}

\def\p{\partial}

\def\bcw{\mathbin{\bigcirc\mkern-15mu\wedge}}

\def\R{\mathcal{R}}

\def\F{\mathcal{F}}

\newtheorem*{Main Theorem}{Main Theorem}
\newtheorem*{Main Proposition}{Main Proposition}

\newtheorem{theorem}{Theorem}[section]
\newtheorem{lemma}[theorem]{Lemma}
\newtheorem{proposition}[theorem]{Proposition}
\newtheorem{corollary}[theorem]{Corollary}
\newtheorem{definition}[theorem]{Definition}
\newtheorem{remark}[theorem]{Remark}
\newtheorem{conjecture}[theorem]{Conjecture}
\numberwithin{equation}{section}

\begin{document}
\title[rigidity of Kerr-Newman] {On the rigidity of stationary charged black holes: small perturbations of the non-extremal Kerr-Newman family}

\author{Li LAI}
\address{Yau Mathematical Sciences Center, Tsinghua University, Beijing, China}
\email{lail14@mails.tsinghua.edu.cn}

\author{Jiong-Yue LI}
\address{Department of Mathematics, Sun Yat-sen University, Guangzhou, China}
\email{lijiongyue@mail.sysu.edu.cn}

\author{Pin Yu}
\address{Yau Mathematical Sciences Center, Tsinghua University, Beijing, China}
\email{yupin@mail.tsinghua.edu.cn}

\thanks{ PY is supported  by NSFC-11522111, NSFC-11825103 and China National Support Program for Young Top-Notch Talents.}

\begin{abstract}
We prove a perturbative result concerning the uniqueness of Kerr-Newman family of black holes: given an asymptotically flat space-time with bifurcate horizons, if it agrees with a non-extremal Kerr-Newman space-time asymptotically flat at infinity and it is sufficiently close to the Kerr-Newman family, then the space-time must be one of the Kerr-Newman solutions. The closeness to the Kerr-Newman family is measured by the smallness of a pair of Mars-Simon type tensors, which were introduced by Wong in \cite{Wong_09} to detect the Kerr-Newmann family.
\end{abstract}

\maketitle

\tableofcontents

\section{Introduction: the statement of the main result}\label{Chapter_Introduction}
In general relativity, the objects under consideration are given by $(\M,g,\Phi)$ where $\M$ is a 4-dimensional smooth manifold, $g$ is a Lorentzian metric on $\M$, and $\Phi$ represents some physical matter fields. The triple $(\M,g,\Phi)$ is said to be a solution if, in addition to the equations of motion for $\Phi$, Einstein's equations are satisfied
\begin{equation*}
Ric - \frac12 R\cdot g = T,
\end{equation*}
where $Ric$ and $R$ are the Ricci and scalar curvatures for the metric $g$ respectively, and $T$ is the energy-momentum tensor determined by $\Phi$ and its derivatives according to the physical theory. In this work, the matter field $\Phi$ will be taken to be an electromagnetic field $H$.

We fix the notations which will be used throughout the paper. A space-time $(\M,g_{ab})$ shall always refer to a four dimensional orientable smooth manifold with a Lorentzian metric $g_{ab}$. A Maxwell field on $\M$ shall always refer to a real 2-form $H_{ab}$ on $\M$ such that it satisfies the Maxwell equations:
 \[\begin{cases}\nabla_{[c}H_{ab]} &= 0,\\
\nabla^aH_{ab} & = 0.
\end{cases}
\]
where $\nabla_a$ is the Levi-Civita connection for the metric $g_{ab}$ and index-raising is done relative to the metric $g_{ab}$. A triplet $(\M, g_{ab}, H_{ab})$ is called an \emph{Einstein-Maxwell space-time} if in addition to Maxwell equations, the Einstein fields equations hold:
\[ R_{ab}-\frac{1}{2}R g_{ab} = T_{ab}.
\]
In the above expressions, $R_{ab}$ and $R$ are the Ricci curvature and scalar curvature  of the metric $g_{ab}$ respectively, $T_{ab}$ is the energy-momentum tensor for the Maxwell field defined as
\begin{equation}\label{energymomentumtensor}
T_{ab} = 2 H_{ac}H_b{}^c - \frac12 g_{ab}H_{cd}H^{cd}.
\end{equation}

Throughout the the paper, we make the assumption that the $(\M, g_{ab}, H_{ab})$ is \emph{stationary}, i.e., there exists a Killing vector field $T^a$ in such a way that
\[\mathcal{L}_T H=0.\]

\bigskip

The most attractive objects in general relativity are black holes. Black hole, roughly speaking, is certain region of the space-time from which light cannot be observed by a far-away observer.  Classical examples of black holes include the Schwarzschild black holes, Kerr family black holes (Schwarzschild black holes are also included in Kerr family and these black holes are vacuum black holes, namely there is no matter field $\Phi$ and $T =0$) and Kerr-Newman family black holes (where the only matter field $\Phi$ is assumed to be an electromagnetic field). The Kerr family describes the geometry of space-time around a
rotating massive body (while Schwarzschild family is just the non-rotating subfamily); the Kerr-Newman family describes charged rotating black holes.

It is believed that, due to gravitational radiation, a generic asymptotically flat solution of the Einstein-Maxwell equations ought to converge asymptotically to a stationary regime. There is also a conjecture which is aimed to understand all the possible asymptotic states of stationary charged black holes:
\begin{conjecture}\label{the_conjecture}
The domains of outer communication of stationary charged black hole solutions are equivalent to those of the Kerr-Newman black holes.
\end{conjecture}
The conjecture, if true, would characterize all possible asymptotic states of the charged black holes. Therefore, we expect the Kerr-Newman family to be the only family of black hole solutions to the Einstein-Maxwell system under reasonable mathematical and physical assumptions. Results of this type are often called \emph{no-hair theorems} in the literature. The present work is an effort to address one aspect of the conjecture.

\medskip

Although one may think black holes are mysterious or even terrible in the real world, mathematically, it enjoys many fascinating properties. For example, Hawking \cite{Hawking_Ellis} showed that near a non-expanding horizon, the space-time has a hidden symmetry, i.e., there is another Killing vector field (which is, usually called Hawking vector field in the literature, different from the stationary vector field) in a neighborhood of the given horizon. The statement is often called Hawking's local rigidity theorem since it only addresses the geometry near horizons.  The existence of such a Hawking vector field plays a fundamental role in the classification theory of stationary black holes. The original construction of this vector field used a restrictive additional assumption, namely one has to assume the space-time is real analytic. Based on solving the characteristic initial value problems for wave equations, Friedrich, R\'{a}cz and Wald \cite{Friedrich_Racz_Wald} constructed the vector field in the smooth
category. Due to the fact that wave equations are only well-posed in the domain of dependence of the horizons, the construction can not be extended to the more physical meaningful domain, namely the domain of outer communication.  Ionescu and Klainerman made a breakthrough on this problem. First of all, they have observed in \cite{Ionescu_Klainerman_Wave} that the Carleman type estimates for wave operators can be used to prove uniqueness theorem for solutions of wave equations out-side of a light-cone. Since the Einstein field equations are wave type equations, they implemented the new idea to study space-times in general relativity. Alexakis, Ionescu and Klainerman proved in \cite{Alexakis_Ionescu_Klainerman} that for Einstein vacuum space-times one can extend the Hawking vector field $K$ to a full neighborhood of the bifurcate sphere without making any additional regularity assumptions. The result has later on been improved by Ionescu and Klaienrman in \cite{Ionescu_Klainerman_JAMS}. Along this line, Ionescu and Klainerman (and later with Alexakis) have proposed a program to study the uniqueness of black holes in a smooth class and not imposing axial symmetry. In \cite{Ionescu_Klainerman_Kerr}, the uniqueness of Kerr family has been proved in the smooth class subject to a technical assumption on the Ernst potential, which can be interpreted as stating that the bifurcate horizon of the given solution is exactly the same as that of the reference Kerr solution. In \cite{Alexakis_Ionescu_Klainerman_Perturbation},  the uniqueness of Kerr family has been proved in the smooth class subject to the condition that the global geometry of the given solution is not too different from that of the reference Kerr solution. In \cite{Alexakis_Ionescu_Klainerman_Duke}, the black hole uniqueness theorem has been proved in the smooth class under the assumption that the bifurcate horizon is close to the special Schwarzschild solutions. Compared to \cite{Ionescu_Klainerman_Kerr}, the main improvement is that the condition on the bifurcate horizon in the present paper is an open condition.  We refer the readers to the survey paper \cite{Ionescu_Klainerman_Survey} of Ionescu and Klainerman for more details.

\medskip

The aforementioned results for the uniqueness of black holes are for the vacuum family of black holes. The local rigidity results, i.e., the existence of Hawking vector field, have also been extend to the electrovacuum case. The work \cite{Yu_local} showed that  one can locally extend the Hawking vector field to a full neighborhood of the bifurcate sphere in the smooth class. Giorgi \cite{Giorgi} proved the existence of the extension of the vector field to a strong null convex domain in an electrovacuum space-time and a result concerning non-extendibility: one can find local, stationary electrovacuum extension of a Kerr–Newman solution in a full neighborhood of a point of the horizon (that is not on the bifurcation sphere) which admits no extension of the Hawking vector field. We remark that these two works only deal with local rigidity and the construction of Hawking vector field near the horizons is merely the first step towards the main Conjecture \ref{the_conjecture} which is a global statement. The second step is to further extend the Hawking vector field to the entire domain of outer communication. This strategy is originally proposed by Hawking: Once we have established the second step, we can appeal to the thesis \cite{Bunting} of Bunting to conclude that our space-time is indeed isometric to one Kerr-Newman black hole, hence the Conjecture \ref{the_conjecture} is proved. The result of Bunting is a generalization of the classical uniqueness theorems of Carter \cite{Carter} and Robinson \cite{Robinson}, it claims that if an asymptotically flat Einstein-Maxwell space-time admits a second rotational Killing vector fields globally, it must be one of Kerr-Newmann black holes.

\medskip

We now can give a rough statement of the main result of the current work: if a stationary charged black hole is locally close to a Kerr-Newman black hole, then it must be a Kerr-Newman black hole. According to above discussions, it suffices to show that the space-time is indeed axially symmetric, i.e., to show the global existence of a second rotational Killing vector field.  We shall use the perturbation strategy developed  in \cite{Alexakis_Ionescu_Klainerman_Perturbation}.  Here, the \emph{closeness to Kerr-Newman} is measured by the
smallness of a pair of tensors $(\mathcal{Q},\mathcal{B})$ which will be defined later in section \ref{section1.2}.  These tensors first defined by W. W-Y. Wong in \cite{Wong_09} are natural generalizations of Mars-Simon tensor. Their size measures the deviation of the space-time from the standard Kerr-Newman family.

\subsection{Non-expanding horizons and local rigidity}\label{double_null_foliation}
This subsection serves as a quick review of the local rigidity results. Let $\mathcal{S}$ be a smoothly embedded space-like $2$-sphere in $\mathcal{M}$ and $\mathcal{H}^+$, $\mathcal{H}^-$ be the corresponding null boundaries of the causal future and the causal past of $\mathcal{S}$. We also assume that both $\mathcal{H}^+$ and $\mathcal{H}^-$ are regular, achronal, null hypersurfaces in a neighborhood $\mathcal{O}$ of $\mathcal{S}$. We shall use the term a
\emph{local regular bifurcate horizon} in $\mathcal{O}$ to denote the triplet $(\mathcal{S}, \mathcal{H}^+, \mathcal{H}^-)$.

To define a double null foliation on the domain $\mathcal{O}\subset \mathcal{M}$ with respect to the triplet
$(\mathcal{S}, \mathcal{H}^+, \mathcal{H}^-)$, one first chooses a smooth future-directed null pair $(L,\Lb)$ along $\mathcal{S}$ normalized as follows
\[ g(L,L) = g(\Lb,\Lb)=0,\quad g(L,\Lb)=-1\]
so that $L$ and $\Lb$  are tangential to $\mathcal{H}^+$ and $\mathcal{H}^-$ respectively. Locally around $\mathcal{S}$, we extend $L$ and $\Lb$ along the null geodesic generators of $\mathcal{H}^+$ and $\mathcal{H}^-$ via parallel transport respectively, i.e., $\nabla_L L = 0$ and $\nabla_{\Lb} \Lb =0$. We turn to the definition of optical functions $u$ and $\ub$. The function $\ub$ (resp. $u$) is defined along $\mathcal{H}^+$ (resp. $\mathcal{H}^-$) by setting initial value $\ub =0$ (resp. $u=0$) on $\mathcal{S}$ and solving $L(\ub)=1$ (resp. $\Lb(u)=1$). Let $\mathcal{S}_{\ub}$ (resp. $\mathcal{S}_u$) be the level surfaces of $\ub$ (resp. $u$) along $\mathcal{H}^+$ (resp. $\mathcal{H}^-$). We define $\Lb$ (resp. $L$) on each point of the hypersurface $\mathcal{H}^+$ (resp. $\mathcal{H}^-$) to be the
unique, future-directed null vector orthogonal to the surface $\mathcal{S}_{\ub}$ (resp. $\mathcal{S}_{u}$) passing through that point and such that $g(L,\Lb)=-1$. The null hypersurface $\mathcal{H}^-_{\ub}$ (resp. $\mathcal{H}^+_u$) is defined to be the set of null geodesics initiating on $\mathcal{S}_{\ub} \subset \mathcal{H}^+$ (resp. $\mathcal{S}_{u} \subset \mathcal{H}^-$) in the direction of $\Lb$ (resp. $L$). We require the null hypersurfaces $\mathcal{H}^-_{\ub}$ (resp. $\mathcal{H}^+_{u}$) to be the level sets of the function $\ub$ (resp. $u$). By this condition, $u$ and $\ub$ are extended to a neighborhood of $\mathcal{S}$ from the null hypersurface $\mathcal{H}^+ \cup \mathcal{H}^-$. Then we can extend both $L$ and $\Lb$ into a neighborhood of $\mathcal{S}$ as gradients of the optical functions
$u$ and $\ub$, i.e., $L=-g^{ab}\p_a u \p_b$ and $\Lb=-g^{ab}\p_a\ub \p_b$. We remark that $g(L,L)=g(\Lb,\Lb)=0$ because $u$ and $\ub$ are null optical functions, while $g(L,\Lb)=-1$ holds (and only in general) on the null surface $\mathcal{H}^+ \cup \mathcal{H}^-$. We define $\mathcal{S}_{u, \ub} = \mathcal{H}^+_u \cap \mathcal{H}^-_{\ub}$. We then choose a local orthonormal frame $e_1, e_2$ along the 2-surface $\mathcal{S}_{u,\ub}$ in such a way that the null frame $\{e_1, e_2, e_3 = \Lb, e_4 =L\}$ satisfies
\[g(e_A,e_B)=\delta_{AB}, \qquad g(e_A,e_3)=g(e_A,e_4)=0,\quad A,B=1,2.\]

With respect to this given null frame, we define two null second fundamental forms $\chi$ and $\chib$ along
$\mathcal{H}^+ \cup \mathcal{H}^-$ as follows:
\[\chi_{AB} = g(\nabla_{e_A}L,e_B), \quad \chib_{AB} = g(\nabla_{e_A}\Lb,e_B).\]
For $\chi$ and $\chib$, their traces are defined by $\tr\chi =
\delta^{AB}\chi_{AB}$ and $\tr\chib = \delta^{AB} \chib_{AB}$; their
traceless parts are denoted by $\chih$ and $\chibh$.

\begin{definition}\label{nonexpansion}
For a local regular bifurcate horizon $(\mathcal{S}, \mathcal{H}^+,
\mathcal{H}^-)$, we say that $\mathcal{H}^+$ is \emph{non-expanding} (resp. $\mathcal{H}^-$) if $\tr\chi=0$ (resp. $\tr\chib=0$) along $\mathcal{H}^+$ (resp. $\mathcal{H}^-$ ). The bifurcate horizon $(\mathcal{S}, \mathcal{H}^+, \mathcal{H}^-)$ is called \emph{non-expanding} if both $\mathcal{H}^+, \mathcal{H}^-$ are non-expanding.
\end{definition}

Let $(\mathcal{S}, \mathcal{H}^+, \mathcal{H}^-)$ be a non-expanding local regular bifurcate horizon in $\mathcal{O}$. The following two theorems will be referred as local rigidity theorems:
\begin{theorem}[\cite{Giorgi},\cite{Yu_local}]\label{local_rigidity_1}Given $(\mathcal{S}, \mathcal{H}^+, \mathcal{H}^-)$ a local, regular, bifurcate, non-expanding horizon in a smooth and time oriented Einstein-Maxwell space-time $(\mathcal{O}, g, H)$, there exists a neighborhood $\mathcal{O}' \subset \mathcal{O}$ of $\mathcal{S}$ and a non-trivial Killing vector field $K$ in $\mathcal{O}'$, which is tangent to the null generators of $\mathcal{H}^+$ and $\mathcal{H}^-$. Moreover, the Maxwell field $H$ is also invariant with respect to $K$, i.e. $\L_K H =0$.
\end{theorem}

As we mentioned, the vector field $K$ is called the \emph{Hawking vector field} in the literature and it was first constructed in the vacuum case in \cite{Hawking_Ellis} assuming the space-time real analytic. In the stationary case, one can indeed construct a second Killing vector field:

\begin{theorem}[\cite{Giorgi},\cite{Yu_local}]\label{local_rigidity_2}
Let $(\mathcal{S}, \mathcal{H}^+, \mathcal{H}^-)$  be a local, regular, bifurcate and non-expanding horizon in a smooth and time oriented Einstein-Maxwell space-time $(\mathcal{O}, g, H)$. If there is a Killing vector field $T$ tangent to $\mathcal{H}^+ \cup \mathcal{H}^-$ which does not vanish identically on $\mathcal{S}$, then there is a neighborhood $\mathcal{O}' \subset \mathcal{O}$ of $\mathcal{S}$, such that we can find a rotational Killing vector $Z$ in $\mathcal{O}'$, i.e. the orbits of $Z$ are closed. Moreover, $[Z,T] = 0$. If in addition $\L_T H =0$, then $\L_Z H =0$.
\end{theorem}

Later on in the proof of global rigidity theorems for stationary space-times, the non-expansion condition will be a consequence of the fact that the Killing vector field $T$ is tangent to $\mathcal{H}^+ \cup \mathcal{H}^-$, see \cite{Hawking_Ellis}. So Theorem \ref{local_rigidity_1} produces a Hawking vector field $K$ in a full
neighborhood of $\mathcal{S}$. The rotational vector field $Z$ provided by Theorem \ref{local_rigidity_2} is a linear
combination of $T$ and $K$, i.e. there exists a constant $\lambda$ such that the one parameter group of diffeomorphism on $\mathcal{M}$ generated by the vector field $Z = T + \lambda K$ is a rotation with a period $T_0$.  We also know that the period $T_0$ is exactly the period of rotations generated by $T$ on the bifurcate sphere $\mathcal{S}$. We refer the readers to \cite{Yu_local} for details.

\subsection{Anti-Self-Duality and Wong's tensors}\label{section1.2}
On the given space-time $(\M,g_{ab})$, the Hodge star operator $*:\wedge^2T^*\M\to\wedge^2T^*\M$ defines a complex structure on the space of (real) 2-forms. In index notation, for a given 2-form $X_{ab}$
\begin{equation}\label{hodgestar}
 {}^*\!X_{ab} = \frac{1}{2}\varepsilon_{abcd} X^{cd}
\end{equation}
where $\varepsilon_{abcd}$ is the volume form.  The complexified space $\wedge^2 T^*\M \otimes_\mathbb{R} \mathbb{C}$ of 2-forms on $\M$ splits into the eigenspaces $\Lambda_{\pm}$ of $*$ with eigenvalues $\pm i$. An element of $\wedge^2 T^*\M \otimes_\mathbb{R} \mathbb{C}$ is \emph{self-dual} (resp. \emph{anti-self-dual}) if it is an eigenvector of $*$ with eigenvalue $i$ (resp. \emph{self-dual}). For a $\mathbb{R}$-valued 2-form
$X_{ab}$, the 2-form
\begin{equation}\label{ASD}
\X_{ab} := \frac{1}{2}(X_{ab} + i \, {}^*\!X_{ab})
\end{equation}
is anti-self-dual, while its complex conjugate $\overline{X}_{ab}$ is
self-dual. In the sequel we shall, in general, write elements of
$\wedge^2 T^*\M$ （$\mathbb{R}$-valued 2-forms） with upper-case Roman letters, and their
corresponding anti-self-dual forms with upper-case calligraphic
letters. In terms of the anti-self-dual forms, the energy-momentum tensor
$T_{ab}$ can be written as
\[T_{ab} = 4 \H_{ac}\Hb_{b}{}^c =4\H_b{}^c\Hb_{ac}.
\]

One can decompose a complex 2-form as a sum of 2-forms in $i$ and $-i$ eigenspaces. This gives the the natural projection operators
\[\P_\pm:\wedge^2 T^*\M \otimes_\mathbb{R}\mathbb{C} \to \Lambda_{\pm}.\]
In terms of local frames, we have
\[(\P_+ X)_{ab} = \bar{\I}_{abcd}X^{cd},\qquad (\P_-X)_{ab} = \I_{abcd}X^{cd},
\]
where
\[ \I_{abcd} = \frac{1}{4}(g_{ac}g_{bd}- g_{ad}g_{bc} + i \varepsilon_{abcd}).
\]
We then define the
\emph{symmetric spinor product} for two anti-self-dual complex two
forms $\X$ and $\mathcal{Y}$:
\[(\mathcal{X}\tilde\otimes\mathcal{Y})_{abcd} :=
\frac12\mathcal{X}_{ab}\mathcal{Y}_{cd} +
\frac12\mathcal{Y}_{ab}\mathcal{X}_{cd} -
\frac13\mathcal{I}_{abcd}\mathcal{X}_{ef}\mathcal{Y}^{ef}
\]
This is a complex $(0,4)$-tensor and it satisfies the algebraic symmetries of a Weyl tensor field:
\begin{itemize}
\item[i)] Anti-symmetric in its first two and last two indices:
$(\mathcal{X}\tilde\otimes\mathcal{Y})_{abcd} = -
(\mathcal{X}\tilde\otimes\mathcal{Y})_{bacd} = -
(\mathcal{X}\tilde\otimes\mathcal{Y})_{abdc}$,
\item[ii)]  Symmetric swapping the first two and the last two sets of
indices: $(\mathcal{X}\tilde\otimes\mathcal{Y})_{abcd} =
(\mathcal{X}\tilde\otimes\mathcal{Y})_{cdab}$,
\item[iii)] The first Bianchi identity: $(\mathcal{X}\tilde\otimes\mathcal{Y})_{abcd} +
(\mathcal{X}\tilde\otimes\mathcal{Y})_{bcad} +
(\mathcal{X}\tilde\otimes\mathcal{Y})_{cabd} = 0$,
\item[iv)] Trace-free: $(\mathcal{X}\tilde\otimes\mathcal{Y})_{abcd}g^{ac} = 0$.
\end{itemize}

Let $C_{abcd}$ is a given Weyl field (instead of usual notation $W_{abcd}$, we shall use $C_{abcd}$ to denote the standard Weyl tensor of the metric $g_{ab}$ in the sequel). We define the left Hodge dual ${}^*C$ and right Hodge dual $C^*$ of
$C$ to be
\[{}^*C_{abcd}  = \frac{1}{2}\varepsilon_{abef}C^{ef}{}_{cd}, \qquad
C^*{}_{abcd}  = \frac{1}{2}C_{ab}{}^{ef}\varepsilon_{efcd}.
\]
We remark that ${}^*C=C^*$ thanks to the trace-free property (they are not equal in general if one defines the left- and right-duals for curvature type tensor). Therefore we can define the \emph{anti-self-dual complexified Weyl tensor} $\C_{abcd}$ of $C_{abcd}$ to be
\begin{equation*}
\C_{abcd} = \frac{1}{2} (C_{abcd} + i\,{}^*C_{abcd}).
\end{equation*}
The anti-self-duality means that $^* \C_{abcd} = -i \,\C_{abcd}$.

Recall that the Riemann curvature tensor $R_{abcd}$ of $g_{ab}$
admits a decomposition
\[ R_{abcd} =  C_{abcd} + \frac{1}{2}(g \bcw S)_{abcd} + \frac{R}{24}(g \bcw g)_{abcd}
\]
where $S_{ab} = R_{ab}-\frac{1}{4}R g_{ab}$ is the traceless part of the Ricci curvature and $\bcw$ is
the Kulkarni-Nomizu product between two $(0,2)$-tensors:
\[(h \bcw k)_{abcd} = h_{ac}k_{bd}+h_{bd}k_{ac}- h_{ad}k_{bc}-h_{bc}k_{ad}.\]
For the Einstein-Maxwell equations, we have $R=0$ and $S_{ab} =R_{ab} =T_{ab}$. We also define
\[\E_{abcd}= \frac{1}{4}\big[(g\bcw T)_{abcd}+i(T_{a}{}^e \varepsilon_{ebcd}+T_b{}^f \varepsilon_{afcd})\big].\]
This tensor together with $\C_{abcd}$ completely characterizes the
Riemann curvature tensor:
\begin{equation*}
 \R_{abcd} = \frac{1}{2}(R_{abcd}+\frac{i}{2} {\varepsilon}_{abef}R^{ef}{}_{cd}) = \C_{abcd}+\E_{abcd}.
\end{equation*}
As a convention, for an arbitrary $(j,k)$-tensors
$Y^{a_1a_2\ldots a_j}_{b_1b_2\ldots b_k}$ and $Z^{a_1a_2\ldots a_j}_{b_1b_2\ldots b_k}$, we write
\[ Y\cdot Z = g_{a_1a_1'}g_{a_2a_2'}\cdots g_{a_ja_j'}g^{b_1b_1'}\cdots
g^{b_kb_k'} Y^{a_1a_2\ldots a_j}_{b_1b_2\ldots b_k}
Z^{a_1'a_2'\ldots a_j'}_{b_1'b_2'\ldots b_k'}. \]
The expression $Z^2$ means $Z\cdot Z$.

\medskip

We turn to the definition of Wong's tensors. This is a generalization of Mars-Simon type tensors for charged stationary space-times.  For the given Killing vector field $T^a$, we can define the Ernst 2-form as
\begin{equation*}
F_{ab} = \nabla_a T_b - \nabla_b T_a = 2\nabla_a T_b.
\end{equation*}
The Ernst 2-form satisfies the Ricci equation
\begin{equation}\label{ricciidentity}
\nabla_c F_{ab} = 2 \nabla_c\nabla_aT_b = 2R_{dcab}T^d.
\end{equation}
The Ernst 2-form also satisfies a divergence-curl system:
\[
\begin{cases}
\nabla_{[c}F_{ab]} & = 0, \\
\nabla^aF_{ab} &= -2 R_{db}T^d.
\end{cases}
\]
According to Cartan's formula, we have
\begin{equation}
0=\L_T \H = \iota_T \circ d\H + d\circ \iota_T \H=d\circ \iota_T \H.
\end{equation}
Thus, the complex-valued 1-form $\H_{ab} T^a$ is closed, since $\M$ is always taken to be simply connected, also exact. In the sequel we will use the complex-valued function
$\Xi$ to denote the associated potential, which is defined by
\begin{equation}\label{EqDefXi}
\nabla_b \Xi = \H_{ab} T^a.
\end{equation}
Notice that \emph{a priori} $\Xi$ is only defined up to a constant. Later on we can use the asymptotic decay \eqref{AFmaxwell} of the Maxwell field to require that $\Xi \to 0$ at spatial infinity. In this way, the potential $\Xi$ is uniquely defined.

We introduce the following four  constants:
\begin{equation}\label{constants}
 C_1 = q_E +i q_B, \qquad C_2 = C_3 = \frac{M}{q_E - i q_B}, \qquad C_4 =|C_2|^2-1 = \frac{M^2}{q^2}-1.
\end{equation}
where $q_E$ and $q_B$ are respectively the electric charge and the magnetic charge of the Maxwell field and $q = \sqrt{q_E^2+q_B^2}$ is the total charge. These constants are well-define once we impose the asymptotic decay \eqref{AFmaxwell}. We now define two fundamental functions
\begin{equation}\label{PandPh}
 P = \frac{1}{2\Xi}, \qquad  \Ph^4 =-\frac{1}{C_1^2 \H^2}.
\end{equation}
Let
\[\widehat{\F}_{ab} = \F_{ab}- 4\Xib \H_{ab}.\]
Tensors $\B_{ab}$ and $\Q_{abcd}$ are defined as follows
\begin{align}\label{BandQ}
C_1^{-1}\B_{ab} &= \widehat{\F}_{ab}+2\Cb_3\H_{ab} = \F_{ab} +(2\Cb_3 -4 \Xib)\H_{ab},\\
\Q_{abcd}&= \C_{abcd}-3\Ph (\F \tilde{\otimes} \H)_{abcd}.\notag
\end{align}
We will refer $\B_{ab}$ and $\Q_{abcd}$ as Wong's tensors.

It straightforward to check that $\widehat{\F}_{ab}$ solves the free Maxwell equation. Therefore, $C_1^{-1}\B_{ab}$ also solves the free Maxwell equations. In particular, it defines a complex potential $V$ with respect to the Killing
vector field $T$, i.e.,
\begin{equation}\label{potentialV}
 \nabla_b V = C_1^{-1} \B_{ab}T^a
\end{equation}
Later on we will require  that $V \rightarrow 0$ at spatial infinity to uniquely define $V$. Finally, we define
the functions $y$ and $z$:
\begin{equation}\label{yandz}
 y =\Re{(C_1 P)} , \ z=\Im{(C_1 P)}
\end{equation}
In other words, $C_1 P =y + iz$.

We end the subsection by a beautiful theorem of Wong. The Mars-Simon tensor was introduced by Mars and Simon, see \cite{Mars_99} and \cite{Mars_00}. It is a tensor defined on
stationary vacuum space-time whose vanishing property locally defines the Kerr family black holes. These new tensors $\B_{ab}$ and $\Q_{abcd}$ were introduced by W. W-Y. Wong in \cite{Wong_09}. Recall that the Kerr-Newman metrics are given by the following complicated expression

 \[-\frac{\Delta- \a^2 \sin^2\theta}{\Sigma}dt^2 -2\a \sin^2\theta \frac{r^2+\a^2-\Delta}{\Sigma}dt d\phi+\frac{(r^2+\a^2)^2-\Delta \a^2 \sin^2\theta}{\Sigma}\sin^2 \theta d\phi^2 + \frac{\Sigma}{\Delta}d r^2 + \Sigma d\theta^2\]
with constants $\Sigma = r^2 + \a^2 \cos^2\theta$ and $\Delta = r^2 -2Mr + \a^2 + q^2$, where $M$, $\a$ and $q$ are the ADM mass, the angular momentum and the charge of the space-time. The vanishing of Wong's tensors $(\Q,\B)$ characterizes the Kerr-Newman family black holes:
\begin{theorem}[W. W-Y. Wong, \cite{Wong_09}] The space-time $(\M,g,H)$ is locally equivalent to one of the Kerr-Newmann black holes if and only if $(\Q,\B)\equiv 0$.
\end{theorem}

\subsection{The global rigidity of Kerr-Newman black holes}\label{Chapter_Global_Rigidity_I}

We assume that there is a given embedded partial Cauchy hypersurface
$\Sigma \subset \M$ which is space-like everywhere. To model the black holes situation, we assume that $\Sigma$ is diffeomorphic to $\mathbb{R}^3 -B_{\frac{1}{2}}$ where $B_{\frac{1}{2}} \subset \mathbb{R}^3$ is a ball of radius $\frac{1}{2}$ centered at the origin.  The diffeomorphism is denoted by $\Phi: \mathbb{R}^3- {B_{\frac{1}{2}}} \rightarrow \Sigma$.
\begin{center}
\includegraphics[width=3in]{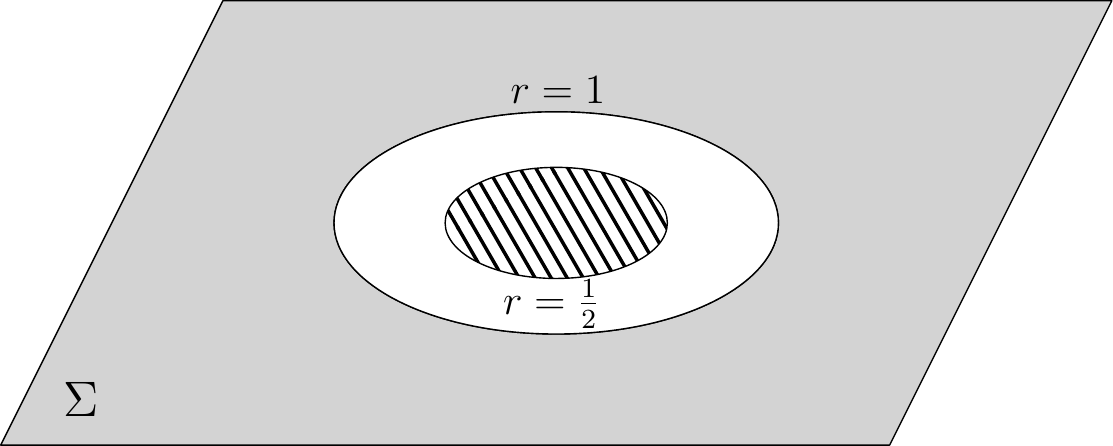}
\end{center}
In the picture, $\Sigma$ is the union of the grey region and the white region. The data will be prescribed on the grey region. We also remark that the diffeomorphism $\Phi$ allows us to use the rectangular coordinate system $(x^1,x^2,x^3)$ on $\Sigma$. We also use the convention that $r=\sqrt{(x^1)^2+(x^2)^2+(x^3)^2}$.

The global assumptions on the space-time and Maxwell field can be grouped into three categories. The first one is the standard asymptotic flatness assumption $\bf{A_1}$. It provides asymptotic expansions of metric $g_{ab}$, Maxwell field $H_{ab}$ and the orbits of the Killing vector field $T$. In particular, it defines the asymptotic region $\M^{end}$ and the domain of outer communication $\E = \I^{-}(\M^{end}) \cap \I^{+}(\M^{end})$ where $\I^{\pm}(\M^{end})$ denotes the future/past set of the set $\M^{end}$. The second assumption $\bf{A_2}$ requires the horizon to be the smooth bifurcate horizon, i.e., the smoothness of two achronal boundaries $\partial (\I^{-}(\M^{end}))$ and $\partial(\I^{+}(\M^{end}))$ in a small neighborhood of their intersection. The intersection $\partial (\I^{-}(\M^{end})) \cap \partial(\I^{+}(\M^{end}))$ is a 2-sphere. This physical assumption is satisfied by most of the known non-extremal family of black holes, e.g., Kerr black holes and Kerr Newmann black holes. The last assumption $\bf{A_3}$
asserts that, in a suitable sense, on the domain of outer communication on $\E$, the space-time as well as the Maxwell field is close to some Kerr-Newman metric. The assumption is made at the level of curvature and it is independent of the choice of coordinates. It requires the smallness of the tensors $\Q_{abcd}$ and $\B_{ab}$. Recall that the vanishing of these tensors characterizes the Kerr-Newman metrics (see \cite{Wong_09}). In Riemannian geometry, if the curvature tensor is sufficiently small (one need to fix the scale), one may conclude that the space is locally close to flat space. The assumption $\bf{A_3}$ is the analogue to this statement.

\begin{center}
\includegraphics[width=3in]{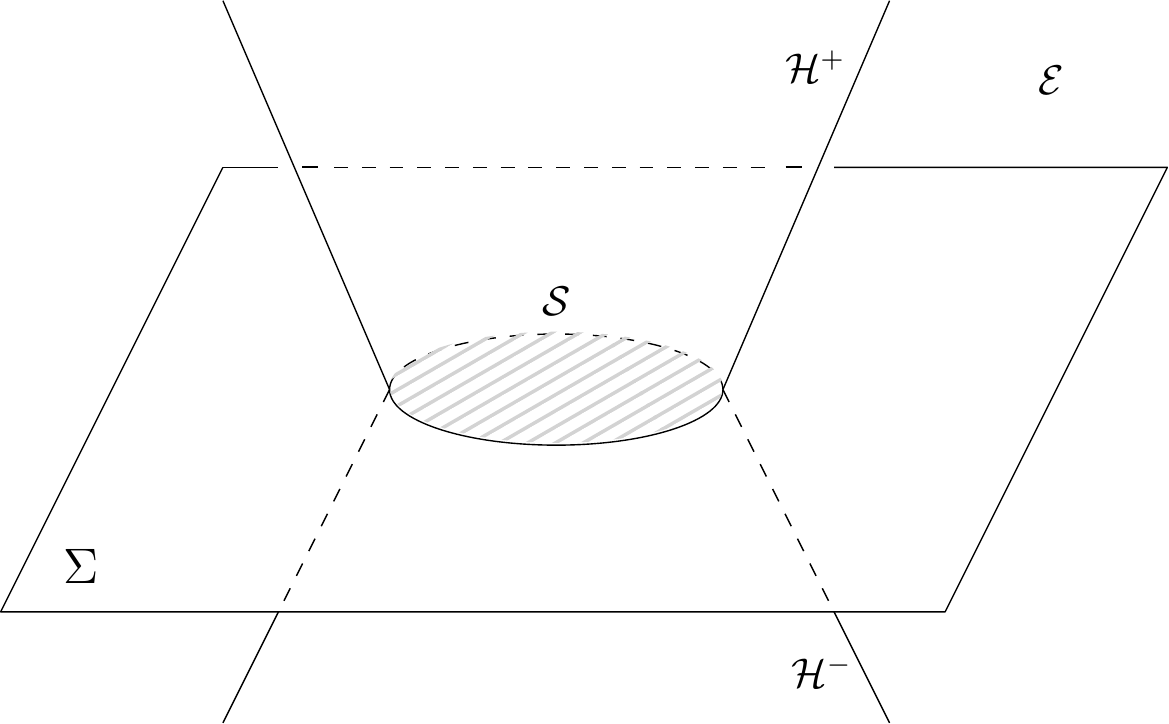}
\end{center}

We now give the precise statements of the assumptions:
\begin{itemize}
\item[($\bf A_1$)] (Asymptotic flatness) We assume that the diffeomorphism $\Phi$ extends to a diffeomorphism
\[\Phi: \mathbb{R} \times (\mathbb{R}^3-B_{R_0}) \rightarrow \M^{end},\]
where $B_{R_0}$ is a ball of radius $R_0$ centered at the origin and $R_0$ is a sufficiently large constant. The asymptotic region $\M^{end}$ is an open set of $\M$. In particular, this diffeomorphism defines a coordinate system $(x^0, x^1, x^2,x^3)$ on $\M^{end}$. We assume that $T = \partial_0$.

In view of the dipole expansions in \cite{Misner_Thorne_Wheeler}, see also \cite{Beig_Simon_80} and \cite{Beig_Simon_81}, we make the the following asymptotic assumptions for the metric and the Maxwell field.

The components of the metric read as\footnote{\label{foot:asympt}We denote by
 $O_k(r^a)$ any smooth function in  $\M^{end}$ which verifies $|\partial^i
f|=O(r^{a-i})$  for any $0\le i\le k$ with $\displaystyle |\partial^i
f|=\sum_{i_0+i_1+i_2+i_3=i}
|\partial_0^{i_0}\partial_1^{i_1}\partial_2^{i_2}\partial_3^{i_3}
f|$; all the repeated indices are understood as Einstein summation convention and $l= 1, 2,3$},
\begin{equation}\label{AFmetric}
\begin{cases}
g_{00}&=-1+\frac{2M}{r}+O_4(r^{-2}),\\
g_{ij}&=\delta_{ij}+\frac{2M}{r}\delta_{ij}+O_4(r^{-2}),\\
g_{0i}&=-\varepsilon_{ijk}\frac{2 S^j x^k}{r^3}+O_4(r^{-3}),
\end{cases}
\end{equation}
where $(S^1,S^2,S^3) \in \mathbb{R}^3$ is the angular momentum
vector and $\varepsilon_{ijk}$ is the fully skew-symmetric 3-tensor;
the components of the Maxwell field read as follows
\begin{equation}\label{AFmaxwell}
\begin{cases}
   { H^{0i} }&=\frac{q x_i}{r^3} + O_4(r^{-3}),\\
   { H^{ij} }&=\frac{q}{M r^3}(\frac{3 S_l x^l}{r^2}x^k-S^k)\varepsilon_{ijk}+O_4(r^{-4}),
\end{cases}
\end{equation}
where $q > 0$ is the total electro-magnetic charge. We define the
total angular-momentum $\mathfrak{a}$ to be
\begin{equation*}
\mathfrak{a^2} =\frac{S^2}{M^2},
\end{equation*}
where $S = \sqrt{(S^1)^2+(S^2)^2+(S^3)^2}$. We require the non-extremal condition
\begin{equation}\label{nonextremal}
 q^2 + \mathfrak{a}^2 < M^2.
\end{equation}
Let $ \E = \I^{-}(\M^{end}) \cap \I^{+}(\M^{end})$. We assume that $\E$ is globally hyperbolic and
\begin{equation}
 \Sigma \cap \I^{-}(\M^{end}) = \Sigma \cap \I^{+}(\M^{end}) = \Phi\big(\mathbb{R}^3 - B_1\big)
\end{equation}
where $\bar{B}_1$ is the ball of radius $1$ centered at the origin. We also assume that the Killing
vector field $T$ is time-like at all points of $\E$ and all orbits of
$T$ in $\E$ are complete and intersect the hypersurface $\Sigma$.

\item[($\bf A_2$)] (Smooth bifurcate horizon) We assume that the intersection of future horizons and past horizons $\partial (\I^{-}(\M^{end})) \cap \partial(\I^{+}(\M^{end}))$ is a 2-sphere $\mathcal{S}$ located \footnote{The assumption that $\mathcal{S}$ lies on  $\Sigma$ will be removed later on by using the local rigidity theorem.} on $\Sigma$.  The future and past event horizons of the space-time are smoothly embedded hypersurfaces. They are defined as
\begin{equation}
 \H^+ = \partial (\I^-(\M^{end})) \qquad \text{and} \qquad \H^- = \partial (\I^+(\M^{end})).
\end{equation}
We assume that $\H^+$ and $\H^-$ are null, non-expanding\footnote{It is well-known that the non-expansion condition is in fact a consequence of the existence of the Killing vector field $T$.} and intersect
transversally at each $\S$. In particular,  this implies that $(\S,\H^+,\H^-)$ is a local,
regular, bifurcate, non-expanding horizon.

Finally, we assume that the Killing symmetry $T$ is tangent to $\H^+$ and
$\H^-$ and does not vanish identically on $\S$.

\item[($\bf A_3$)] (Perturbation of a Kerr-Newman black hole) Let $\B_{ab}$ and $\Q_{abcd}$ be the Mars-Simon type tensors. We assume that on $\Sigma$ we have
\[ \Big|\frac{\Q_{abcd}}{2\Xib - C_3}\Big| + \Big|\frac{\B_{ab}}{2\Xib - C_3}\Big|+\Big|\frac{\nabla_c \B_{ab}}{2\Xib - C_3}\Big| < \varepsilon,
\]
where $\varepsilon$ is a small constant to be determined in the proof and $C_3 = \frac{M}{q_E -iq_B}$.
\end{itemize}

We remark that in the assumption ${\bf A_1}$, we can allow more than one black hole and show that the multi-black-hole configuration is ruled out by the assumption ${\bf A_3}$. This result is proved in \cite{Wong_Yu} and it should be regarded as a generalization of a celebrated theorem of Bunting and Masood-ul-Alam (see \cite{Bunting_Masood-ul-Alam}) which asserts that there is no mutli-black-hole configuration for static space-times.

\bigskip

We now state the main result of the paper:

\begin{Main Theorem}\label{Main_Theorem}
Under the assumptions $\bf A_1$, $\bf A_2$ and $\bf A_3$, there exists a constant $\varepsilon_0>0$, for all $\varepsilon<\varepsilon_0$ in $\bf A_3$, the domain of outer communication $\E$ of $\M$ is isometric to the domain of outer communication of the Kerr-Newman space-time with mass $M$, angular momentum $(S^1,S^2,S^3)$ and charge $q$.
\end{Main Theorem}

As we mentioned before, the proof relies on Hawking's original strategy: we reduce the case of general stationary space-times to that of stationary and axi-symmetric space-times for which we can use Bunting's uniqueness theorem, see \cite{Bunting}. Thus, it suffices to prove the following proposition:
\begin{Main Proposition}\label{Main_Theorem_H}
Under the assumptions $\bf A_1$, $\bf A_2$ and $\bf A_3$,  there exists a constant $\varepsilon_0>0$, for all $\varepsilon<\varepsilon_0$ in $\bf A_3$, there is a non-trivial smooth Killing vector field $Z$ on the entire domain of outer communication $\E$ such that the following properties hold:
\begin{itemize}
\item[(1)] $Z$ is tangent to $\H^+ \cup \H^-$;
\item[(2)] $\L_Z H=0$, $[Z,T] = 0$ and $Z(y)=0$;
\item[(3)] There is a $\tau_0 >0$ such that $\phi_{\tau_0} = {\rm Id}$ on $\E$, where $\phi_\tau$ is the flow generated by $Z$.
\end{itemize}
\end{Main Proposition}
Our approach is adapted from \cite{Alexakis_Ionescu_Klainerman_Perturbation} where the authors proved a similar result for vacuum space-times. We first prove that the function $y$ satisfies the $T$-conditional pseudo-convexity property away from the horizon. This pseudo-convexity property, which was first observed in \cite{Ionescu_Klainerman_Kerr} in the case of the Kerr space-times, plays a key role in the Carleman estimates and the uniqueness argument. Then we extend the Hawking vector field $K$ constructed locally near the horizon from Theorem \ref{local_rigidity_1} to the entire domain of outer communication $\E$. The construction of $Z$ then follows.

\section{Geometric preparations: a Mars type lemma}
The main goal of the present section is to derive a Mars type Lemma for Wong's pairs $(\B_{ab},\Q_{abcd})$. It is already proved in \cite{Wong_09} when one assumes the vanishing of $\B_{ab}$ and $\Q_{abcd}$ by using the Newman-Penrose type formalism. The current situation is much more complicated due to the fact that we have a lot of error terms coming from $\B_{ab}$ and $\Q_{abcd}$.  We shall proceed in the tensorial way and write down the precise \emph{algebraic} expressions of all the error terms. This serves as the basic tool to run the bootstrap argument in the next section. In the sequel, when we say a term is an \emph{error term} or an \emph{algebraic error term}, it means that if we assume $\B_{ab}$, $\nabla \B_{ab}$ and $\Q_{abcd}$ are considerably small (in $L^{\infty}$ norm), then this term is small. In particular, if the space-time is locally a Kerr-Newman solution, the algebraic error terms are identically zero.

\medskip

We start with the following lemma, where the functions $P$ and $\Ph$ are defined in \eqref{PandPh}:
\begin{lemma}\label{lemmaPPh}
We have the following algebraic identity:
\begin{align*}
4C_1^{-1}(\Cb_3-2\Xib)\nabla_b(\frac{1}{P}-\frac{1}{\Ph})&= \Ph^3\H \cdot \nabla_b \B-2C_1\Ph^3 \Q_{dbef}\H^{ef}T^d + 3 C_1^{-2}\B_{db}T^d-\Ph^4 (\B\cdot \H) \H_{db}T^d.
\end{align*}
\end{lemma}
\begin{proof}
 We start by calculating
\begin{align*}
\H^{ab}\nabla_c\F_{ab} & = 2(\C_{dcab}+\E_{dcab})T^d\H^{ab}\\
&=2[ \Q_{dcab} + 3\Ph(\F\tensor\H)_{dcab}]T^d\H^{ab} + 2(\T_{ad}\H^a{}_c + \T_{bc}\H_d{}^b)T^d \\
&= 2\Q_{dcab}\H^{ab}T^d + \Ph (3\F_{dc}\H_{ab}\H^{ab} + \H_{dc}\F_{ab}\H^{ab})T^d  + 8(\H_{af}\Hb_d{}^f\H^a{}_c + \H_{bf}\Hb_c{}^f\H_d{}^b)T^d \\
& = 2\Q_{dcab}\H^{ab}T^d + \Ph(3[C_1^{-1}\B_{dc}- (2\Cb_3-4\Xib)\H_{dc}]\H_{ab}\H^{ab} \\
& \quad + \H_{dc}[C_1^{-1}\B_{ab} - (2\Cb_3 - 4\Xib)\H_{ab}]\H^{ab})T^d + 4\H_{ab}\H^{ab}\Hb_{dc}T^d\\
&= 2\Q_{dcab}\H^{ab}T^d-3C_1^{-3}\Ph^{-3}\B_{dc}T^d + 4C_1^{-2}\Ph^{-3}(2\Cb_3-4\Xib)\H_{dc}T^d\\
&\quad+ C_1^{-1}\Ph \B\cdot \H \H_{dc}T^d-4C_1^{-2}\Ph^{-4}\Hb_{dc}T^d.
\end{align*}
On the other hand, we can compute
\begin{align*}
 \H^{ab}\nabla_c[(2\Cb_3-4\Xib)\H_{ab}]&=4C_1^{-2}\Ph^{-4}\Hb_{cd}T^d -4\Ph^{-3}C_1^{-2}(\Cb_3-2\Xib)\nabla_c\frac{1}{\Ph}.
\end{align*}
By adding them together, we have
\begin{align*}
 C^{-1}\H^{ab}\nabla_c\B_{ab} &= 4\Ph^{-3}C_1^{-2}(\Cb_3 -2\Xib)(2\nabla_c\Xi -\nabla_c\frac{1}{\Ph})\\
&\qquad +2\Q_{dcab}\H^{ab}T^d -3C_1^{-3}\Ph^{-3}\B_{dc}T^d + C_1^{-1} \Ph(\B\cdot\H) \H_{dc}T^d,
\end{align*}
which proves the lemma.
\end{proof}
\begin{lemma}\label{lemmaTsquare}
 We have the following identities:
\begin{align*}
 T^2 &= -|\frac{1}{P}-C_2|^2 + C_4 - 2\Re(V), \\
\nabla(C_1 P) \cdot \nabla(C_1 P) &= -\frac{P^4}{\Ph^4}T^2, \\
\Box_g (C_1 P) &= -\frac{2(\Cb_1 C_2-\Cb_1\Pb)}{C_1 P \Cb_1 \Pb}-(\frac{1}{P^4}-\frac{1}{\Ph^4})\frac{2P^3}{C_1}(\frac{C_2}{\Pb}-1)+ P^2 \H \cdot \B -4C_1P^3 \H^2\Re(V).
\end{align*}
\end{lemma}
\begin{proof}
 For the first identity, we compute as follows
\begin{align*}
 \nabla_a T^2 &= -F_{ba}T^b = -2\Re(\F_{ba}T^b)\\
&=-2\Re[C_1^{-1}\B_{ba}T^b -(2\Cb_3-4\Xib)\H_{ba}T^b]\\
&=-2\Re(\nabla_a V)+4\Re[(2\Xib -\Cb_3)\nabla_a \Xi]\\
&=-2\Re(\nabla_a V)-\nabla_a |\frac{1}{P}-C_2|^2.
\end{align*}
This shows the first identity up to a constant $C_4$.
\begin{remark}
In view of the assumption $\bf A_1$, it is straightforward to derive $C_4 = |C_2|^2-1 =\frac{M^2}{q^2}-1$ by studying the asymptotics at spatial infinity.
\end{remark}
The second identity is a direct computation and we omit the proof.

Now we turn to the third one
\begin{align*}
 \Box_g P &=-2\nabla^a (P^2 \H_{ba}T^b) \\
 &= 8P^3\H_{c}{}^{a}T^c\H_{ba}T^b+ P^2 \H_{ab}\F^{ab}\\
          &=2P^3\H^2 T^2 + P^2\H^{ab}[C_1^{-1}\B_{ab}-(2\Cb_3-4\Xib)\H_{ab}]\\
      &=2P^2 \H^2[T^2 P+(\frac{1}{\Pb}-\Cb_3)]+C_1^{-1} P^2 \H \cdot \B\\
          &=2P^2 \H^2[( -|\frac{1}{P}-C_2|^2 + C_4 - 2\Re(V))P+(\frac{1}{\Pb}-\Cb_3)]+C_1^{-1} P^2 \H \cdot \B\\
      &=2P^2 \H^2[( -|\frac{1}{P}-C_2|^2 + C_4 )P+(\frac{1}{\Pb}-\Cb_3)]+C_1^{-1} P^2 \H \cdot \B -4P^3 \H \cdot\H\Re(V) \\
      &=-\frac{2P^2}{C_1^2\Ph^4}(\frac{C_2 P}{\Pb}+(C_4-|C_2|^2)P)+C_1^{-1} P^2 \H \cdot \B -4P^3 \H \cdot\H\Re(V) \\
      &=-\frac{2P^3}{C_1^2\Ph^4}(\frac{C_2}{\Pb}-1)+C_1^{-1} P^2 \H \cdot \B -4P^3 \H \cdot\H\Re(V) \\
      &=-\frac{2(\Cb_1 C_2-\Cb_1\Pb)}{C_1^2 P \Cb_1 \Pb}-(\frac{1}{P^4}-\frac{1}{\Ph^4})\frac{2P^3}{C_1^2}(\frac{C_2}{\Pb}-1)+C_1^{-1} P^2 \H \cdot \B -4P^3 \,\H \cdot\H\,\Re(V).
\end{align*}
This completes the proof of the lemma.
\end{proof}
\begin{corollary}\label{cornablas}We also have the following identities:
\begin{equation}\label{nablayandz}
\nabla y \cdot \nabla z =\frac{1}{2} \Im ((1-\frac{P^4}{\Ph^4})T^2), \quad (\nabla y)^2 -(\nabla z)^2 =-T^2+\Re ((1-\frac{P^4}{\Ph^4})T^2),
\end{equation}
\begin{equation}\label{boxy}
\Box_g y = -\frac{2(M-y)}{y^2+z^2}+e(y), \quad \Box_g z = -\frac{2z}{y^2+z^2}+e(z).
\end{equation}
where the algebraic error terms $e(y)$ and $e(z)$ are given by
\begin{align*}
 e(y) &= \Re[-(\frac{1}{P^4}-\frac{1}{\Ph^4})\frac{2P^3}{C_1}(\frac{C_2}{\Pb}-1)+ P^2 \H \cdot \B -4C_1P^3 \H^2 \Re(V)],\\
 e(z) &= \Im[-(\frac{1}{P^4}-\frac{1}{\Ph^4})\frac{2P^3}{C_1}(\frac{C_2}{\Pb}-1)+ P^2 \H \cdot \B -4C_1P^3 \H^2 \Re(V)].
\end{align*}
\end{corollary}
In the remaining part of this section, we decompose various geometric objects with respect to a naturally defined null frame adapted to the Maxwell field. Because it captures the Lorentian character of the theory,  when one computes, all kinds of cancelations can be easily observed via this decomposition.

Since $\H_{ab}$ is an anti-self-dual 2-form, if we assume $\H^2 \neq 0$ (this condition will be automatically true when run bootstrap argument in the next section, in particular, the consequence of the bootstrap argument will show that $\H^2 \neq 0$ on the whole $\Sigma$. So in what follows, we always assume this condition), then it has two distinct future directed principal null directions, which we denote by $\lb^a$ and $l^a$, with the normalization $g(\lb,l)=-1$. Thus $\H_{ab}$ can be written as
\begin{equation}\label{nullH}
 \H_{ab} = \frac{1}{2C_1 \Ph ^2}(\lb_a l_b -l_a\lb_b + i\varepsilon_{abcd}\lb^c l^d).
\end{equation}
Recall that $P = \frac{1}{2\Xi}$, together with the definition of $\Xi$, this implies
\begin{equation*}
 2C_1 P^2 \H_{ab}T^a =-C_1 \nabla_b P = -(\nabla_b y + i \nabla_b z),
\end{equation*}
which implies
\begin{equation}\label{nullnablayz}
 \nabla_b y = (l\cdot T)\lb_b -(\lb \cdot T) l_b + E_b(y), \quad \nabla_b z = \varepsilon_{bacd}T^a \lb^c l^d +E_b(z)
\end{equation}
with the algebraic error terms defined as follows
\begin{align}\label{errornablayz}
 E_b(y) &= \Re\big[(1-\frac{P^2}{\Ph^2})\big((\lb\cdot T) l_b -(l\cdot T)\lb_b + i\varepsilon_{abcd}T^a\lb^c l^d\big)\big],\notag\\
 E_b(z) &= \Im\big[(1-\frac{P^2}{\Ph^2})\big((\lb\cdot T) l_b -(l\cdot T)\lb_b + i\varepsilon_{abcd}T^a\lb^c l^d\big)\big].
\end{align}

We are going to derive a key lemma in this section. The first
version of the lemma appears in \cite{Mars_99} for stationary vacuum
space-time. Another version of the lemma is proved by W. W-Y. Wong in
\cite{Wong_09} where one assumes the vanishing of $\B_{ab}$ and
$\Q_{abcd}$. The following covariant proof is inspired by the work of Alexakis, Ionescu and Klainerman in \cite{Alexakis_Ionescu_Klainerman_Perturbation}.

\begin{lemma}\label{algebraicmars}
 Let $U = (y^2 + z^2)\nabla^a z \nabla_a z + z^2$, then the gradient $\nabla_b U$ is an algebraic error term. The precise
algebraic expression of $\nabla_b U$ will be given in the proof.
\end{lemma}
\begin{proof}
We compute $\nabla_b U$ and divide it into three terms:
\begin{align*}
 \nabla_b U &= \underbrace{(2y\nabla_b y  + 2z\nabla_b z)\nabla^a z \nabla_a z}_{{I_1}} + \underbrace{2z\nabla_b z}_{{I_2}} + \underbrace{2|C_1 P|^2\nabla^a z \nabla_a\nabla_b z}_{{I_3}}.
\end{align*}
The most difficult term is $I_3$. Since $y$ is the real part of $C_1 P$, we will compute $\nabla_a \nabla_b(C_1 P)$ and then take the real part to derive the expression for $T_3$. Thus,
\begin{align*}
 \nabla_a \nabla_b(C_1 P) &= C_1 \nabla_a\nabla_b(\frac{1}{2\Xi})\\
 &=-2 C_1 \nabla_a(P^2 \nabla_b \Xi)\\
 &=\frac{2C_1}{P}\nabla_a P \nabla_b P-2C_1 P^2 \nabla_a(\H_{cb} T^c).
\end{align*}
Now we use the definition of $\B_{ab}$ to convert $\H_{ab}$ in the
second term into terms involving $\B_{ab}$ and $\F_{ab}$. In what follows, as long as a term involves $\B_{ab}$ or $\Q_{abcd}$, it will be treated as an error term. As to the
$\F_{ab}$ part, by virtue of Ricci identity \eqref{ricciidentity},
one can create a curvature term through $\nabla F_{ab}$:
\begin{align*}
&\quad \nabla_a \nabla_b(C_1 P)\\
 &=\frac{2C_1}{P}\nabla_a P \nabla_b P- C_1 P^2 \H_{cb} F_a{}^c -2 C_1 P^2 T^c \nabla_a (\frac{\F_{cb}-C_1^{-1}\B_{cb}}{4\Xib - 2\Cb_3})\\
&= \frac{2C_1}{P}\nabla_a P \nabla_b P- C_1 P^2 \H_{cb}{ (\F_a{}^c+\overline{\mathcal{F}}_a{}^c) }+8 C_1 P^2 T^c \frac{\nabla_a \Xib}{(4\Xib - 2\Cb_3)^2}(\F_{cb}-C_1^{-1}\B_{cb})\\
&\quad - C_1 P^2 T^c \frac{\nabla_a \F_{cb}}{2\Xib - \Cb_3} + P^2 T^c \frac{\nabla_a \B_{cb}}{2\Xib - \Cb_3}.
\end{align*}
Now we try to write things in terms of $\H$ and $P$, and treat the
terms involving $\Q$ and $\B$ as error terms:
\begin{align*}
 & \quad \nabla_a \nabla_b(C_1 P)\\
 &=\frac{2C_1}{P}\nabla_a P \nabla_b P- C_1 P^2 \H_{cb} (C_1^{-1} \B_a{}^c-(2\Cb_3 -4\Xib)\H_a{}^c)-C_1 P^2 \H_{cb} {(\overline{C}_1^{-1} \Bb_a{}^c-(2C_3 -4\Xi)\Hb_a{}^c)}\\
 & +4 C_1 P^2 T^c \frac{\H_{cb} \nabla_a \Xib}{2\Xib - \Cb_3}- 2C_1 P^2 T^c  \frac{(\C_{dacb}+\E_{dacb})T^d}{2\Xib - \Cb_3} + P^2 T^c \frac{\nabla_a \B_{cb}}{2\Xib - \Cb_3}.
\end{align*}
Thus,
\begin{align*}
 &\quad \nabla_a \nabla_b(C_1 P)\\
 &=\frac{2C_1}{P}\nabla_a P \nabla_b P- 2C_1 P^2 [\H_{cb} \H_a{}^c(2\Xib -\Cb_3)+\H_{cb} \Hb_a{}^c(2\Xi -C_3)]+\frac{C_1 \nabla_a \Pb \nabla_b P}{(2\Xib - \Cb_3)\Pb^2} - \frac{2C_1 P^2}{2\Xib - \Cb_3}\E_{dacb} T^c T^d \\
&\quad-\frac{6C_1 P^2\Ph}{2\Xib - \Cb_3}(\F \tensor \H)_{dacb} { T^c T^{d} }+ P^2 T^c \frac{\nabla_a \B_{cb}}{2\Xib - \Cb_3}-\frac{2C_1 P^2}{2\Xib - \Cb_3}\Q_{dacb} T^c T^d-2C_1 P^2 \H_{cb}\Re(C_1^{-1}\B_{a}{}^c).
\end{align*}
Finally, we obtain
\begin{align*}
\nabla_a \nabla_b(C_1 P)
&=\frac{2C_1 }{P}\nabla_a P \nabla_b P- 2C_1 P^2 \H_{cb} \Big[\H_a{}^c(2\Xib -\Cb_3)+\Hb_a{}^c(2\Xi -C_3)\Big]	\\
&\quad +\frac{C_1 \nabla_a \Pb \nabla_b P}{(2\Xib - \Cb_3)\Pb^2}- \frac{2C_1 P^2}{2\Xib - \Cb_3}\E_{dacb} T^c T^d  -12C_1 P^3(\H \tensor \H)_{dacb} T^c T^d\\
&\quad + P^2 T^c \frac{\nabla_a \B_{cb}}{2\Xib - \Cb_3}-\frac{2C_1 P^2}{2\Xib - \Cb_3}\Q_{dacb} T^c T^d + 12C_1P^2(P-\Ph)(\H\tensor\H)_{dacb}T^c T^d\\
&\quad-2C_1 P^2 \H_{cb}\Re(C_1^{-1}\B_{a}{}^c)-\frac{6 P^2\Ph}{2\Xib - \Cb_3}(\B \tensor \H)_{dacb} T^c T^d.
\end{align*}
We define the first error term
\begin{align*}
 E(1)_{ab} &= P^2 T^c \frac{\nabla_a \B_{cb}}{2\Xib - \Cb_3}-\frac{2C_1 P^2}{2\Xib - \Cb_3}\Q_{dacb} T^c T^d+ 12C_1P^2(P-\Ph)(\H\tensor\H)_{dacb}T^c T^d\\
&\quad-2C_1 P^2 \H_{cb}\Re(C_1^{-1}\B_{a}{}^c)-\frac{6 P^2\Ph}{2\Xib - \Cb_3}(\B \tensor \H)_{dacb} T^c T^d.
\end{align*}
By the previous computation, we have
\begin{align*}
\nabla_a \nabla_b(C_1 P) &= \underbrace{\frac{2C_1 }{P}\nabla_a P \nabla_b P}_{S(1)_{ab}} \underbrace{-12C_1 P^3(\H \tensor \H)_{dacb} T^c T^d}_{S(2)_{ab}} +\underbrace{\frac{C_1 \nabla_a \Pb \nabla_b P}{(2\Xib - \Cb_3)\Pb^2}}_{S(3)_{ab}} \underbrace{-\frac{2C_1 P^2}{2\Xib - \Cb_3}\E_{dacb} T^c T^d}_{S(4)_{ab}} \\
&\quad \underbrace{-2C_1 P^2 \H_{cb} [\H_a{}^c(2\Xib -\Cb_3)+\Hb_a{}^c(2\Xi -C_3)]}_{S(5)_{ab}}  +E(1)_{ab}\\
&=[S(1)_{ab}+S(2)_{ab}] + [S(3)_{ab}+ S(4)_{ab}] + S(5)_{ab} + E(1)_{ab}.
\end{align*}
Now we go further to compute these $S(i)_{ab}$'s one by one as follows:
\begin{align*}
 S(2)_{ab} &= -12C_1 P^3(\H \tensor \H)_{dacb} T^c T^d\\
&=-12C_1 P^3(\H_{da}\H_{cb} -\frac{1}{3}\I_{dacb}\H \cdot \H) T^c T^d\\
&=  -\frac{3C_1 }{P}\nabla_a P \nabla_b P + C_1 P^3 \H\cdot \H(g_{dc}g_{ab}-g_{db}g_{ac} +i \varepsilon_{dacb})T^cT^d \\
&=-\frac{3C_1 }{P}\nabla_a P \nabla_b P -\frac{1}{C_1}\frac{P^3}{\Ph^4}T^2 g_{ab} + \frac{1}{C_1} \frac{P^3}{\Ph^4}T_a T_b.
\end{align*}
So we have
\begin{align*}
 S(1)_{ab}+S(2)_{ab} = -\frac{C_1 }{P}\nabla_a P \nabla_b P -\frac{T^2}{C_1 P}g_{ab}+ \frac{T^2}{C_1}(\frac{1}{P}-\frac{P^3}{\Ph^4})g_{ab} + \frac{1}{C_1} \frac{P^3}{\Ph^4}T_a T_b.
\end{align*}
Now compute $S(4)_{ab}$:
\begin{align*}
 S(4)_{ab} &= -\frac{2C_1 P^2}{2\Xib - \Cb_3}\E_{dacb} T^c T^d\\
       &= -\frac{C_1 P^2}{4\Xib - 2\Cb_3}[g_{dc}T_{ab}+g_{ab}T_{dc}-g_{db}T_{ac}-g_{ac}T_{db}+i (T_d{}^e \varepsilon_{eacb}+T_a{}^f \varepsilon_{dfcb})] T^c T^d\\
       &= -\frac{C_1 P^2}{4\Xib - 2\Cb_3}(T^2T_{ab}+ g_{ab}T_{dc}T^c T^d-T_b T_{ac}T^c-T_a T_{db}T^d+i T^dT_{d}{}^e\varepsilon_{eacb}T^c)\\
           &= -\frac{2C_1 P^2}{2\Xib - \Cb_3}(T^2 \H_{ae}\Hb_b{}^{e} + g_{ab}\H_{de}T^d\Hb_{c}{}^e T^c -T_b \Hb_a{}^e \H_{ce}T^c-T_a T_{db} T^d +i \H_{df}T^d \Hb^{ef}\varepsilon_{eacb}T^c)\\
       &= -\frac{2C_1 P^2}{2\Xib - \Cb_3}(T^2 \H_{ae}\Hb_b{}^{e}+ \frac{\nabla P \cdot \nabla \Pb}{4P^2\Pb^2}g_{ab} + { \frac{1}{2P^2} T_b \Hb_a{}^e \nabla_e P-T_a T_{db} T^d) }-\frac{2C_1 P^2 i}{2\Xib - \Cb_3}\nabla_f \Xi \Hb^{ef}\varepsilon_{eacb}T^c.
\end{align*}
By virtue of the identity $i \Hb^{hk}\varepsilon_{wyzk} = -3g^h{}_{[w}\Hb_{yz]}$, we have
\begin{align*}
 S(4)_{ab} &= -\frac{2C_1 P^2}{2\Xib - \Cb_3}(T^2 \H_{ae}\Hb_b{}^{e}+ \frac{\nabla P \cdot \nabla \Pb}{4P^2\Pb^2}g_{ab} +\frac{1}{2P^2} T^b \Hb_a{}^e \nabla_e P-T_a T_{db} T^d)\\
       &\quad +{ \frac{2iC_1 P^2}{2\Xib - \Cb_3} }
          \nabla^f \Xi (g_{fa}\Hb_{cb}+g_{fc}\Hb_{ba}+g_{fb}\Hb_{ac})T^c\\
       &= -\frac{2C_1 P^2}{2\Xib - \Cb_3}(T^2 \H_{ae}\Hb_b{}^{e}-T_a T_{db} T^d)-\frac{C_1 \nabla P \cdot \nabla \Pb}{(4\Xib - 2\Cb_3)\Pb^2}g_{ab}-\frac{C_1}{2\Xib-\Cb_3}\Hb_a{}^e\nabla_e P T_b\\
       &\quad +\frac{2C_1 P^2}{2\Xib - \Cb_3}(\nabla_a \Xi \nabla_b \Xib -\nabla_b \Xi \nabla_a \Xib)\\
       &= -\frac{2C_1 P^2}{2\Xib - \Cb_3}(T^2 \H_{ae}\Hb_b{}^{e}-T_a T_{db} T^d) -\frac{C_1}{2\Xib-\Cb_3}\Hb_a{}^e\nabla_e P T_b\\
       &\quad +\frac{C_1}{(4\Xib - 2\Cb_3)\Pb^2}(\nabla_a P \nabla_b \Pb -\nabla_b P \nabla_a \Pb-\nabla P \cdot \nabla \Pb g_{ab}).
\end{align*}
This implies
\begin{align*}
S(3)_{ab}+S(4)_{ab} &= -\frac{2C_1 P^2}{2\Xib - \Cb_3}(T^2 \H_{ae}\Hb_b{}^{e}-T_a T_{db} T^d) -\frac{C_1}{2\Xib-\Cb_3}\Hb_a{}^e\nabla_e P T_b\\
       &\quad +\frac{C_1}{(4\Xib - 2\Cb_3)\Pb^2} {(3\nabla_a P \nabla_b \Pb + 3\nabla_b P \nabla_a \Pb-\nabla P \cdot \nabla \Pb g_{ab}) }.
\end{align*}
In conclusion, one can write
\begin{align*}
\nabla_a \nabla_b(C_1 P) &=[S(1)_{ab}+S(2)_{ab}] + [S(3)_{ab}+ S(4)_{ab}] + S(5)_{ab} + E(1)_{ab}\\
&= -\frac{C_1 }{P}\nabla_a P \nabla_b P -\frac{T^2}{C_1 P}g_{ab} -\frac{2C_1 P^2}{2\Xib - \Cb_3}T^2 \H_{ae}\Hb_b{}^{e} -\frac{C_1}{2\Xib-\Cb_3}\Hb_a{}^e\nabla_e P T_b\\
       &\quad +\frac{C_1}{(4\Xib - 2\Cb_3)\Pb^2}(\nabla_a P \nabla_b \Pb +\nabla_b P \nabla_a \Pb-\nabla P \cdot \nabla \Pb g_{ab})\\
       &\quad -2C_1 P^2 \H_{cb} [\H_a{}^c(2\Xib -\Cb_3)+\Hb_a{}^c(2\Xi -C_3)]\\
&\quad + E(1)_{ab}+\frac{T^2}{C_1}(\frac{1}{P}-\frac{P^3}{\Ph^4})g_{ab} + \frac{1}{C_1} \frac{P^3}{\Ph^4}T_a T_b + \frac{2C_1 P^2 T_{cb} T^c}{2\Xib - \Cb_3}T_a.
\end{align*}
In view of the expression of $I_3$, one needs to multiply the
previous expression by the factor $\nabla^a z$, the last two terms
will drop off since $\nabla_T z =0$ which comes from the symmetry of
the space-time and Maxwell field. To simplify, we define the
second error term
\begin{equation} \label{E(2)}
 E(2)_{ab}= E(1)_{ab} + \frac{T^2}{C_1}(\frac{1}{P}-\frac{P^3}{\Ph^4})g_{ab}.
\end{equation}
Thus,
\begin{align*}
 &\quad 2|C_1 P|^2\nabla^a z \nabla_a \nabla_b (C_1 P) \\
 &= -2C_1^2 \Cb_1 \Pb \nabla_a P \nabla_b P \nabla^a z -\frac{2C_1^2 \Cb_1 P\Pb}{2\Xib -\Cb_3}\Hb_a{}^e \nabla^a z \nabla_e P T_b- 2T^2\Cb_1 \Pb \nabla_b z\\
 &\quad -\frac{4C_1^2 \Cb_1 P^3 \Pb}{2\Xib - \Cb_3}T^2 \H_{ae}\Hb_b{}^{e} \nabla^a z  + \frac{C_1^2 \Cb_1 P}{(2\Xib - \Cb_3)\Pb}(\nabla_a P \nabla^a z\nabla_b \Pb +\nabla_b P \nabla_a \Pb \nabla^a z-\nabla P \cdot \nabla \Pb \nabla_b z)\\
&\quad -4C_1^2 \Cb_1 P^3 \Pb \H_{cb} [\H_a{}^c(2\Xib -\Cb_3)+\Hb_a{}^c(2\Xi -C_3)]\nabla^a z+ 2|C_1 P|^2 E(2)_{ab}\nabla^a z\\
&=R(1)_b+R(2)_b+R(3)_b+R(4)_b+R(5)_b+R(6)_b+2|C_1 P|^2 E(2)_{ab}\nabla^a z.
\end{align*}
With the help of the null decompositions, we now show that $R(2)_b$
is actually an error term.
\begin{align*}
 R(2)_b &= -\frac{2C_1^2 \Cb_1 P\Pb}{2\Xib -\Cb_3}\Hb_a{}^e \nabla^a z \nabla_e P T_b = -\frac{2|C_1 P|^2}{2\Xib -\Cb_3}\Hb_a{}^e \nabla^a z \nabla_e (C_1 P) T_b\\
&\stackrel{\Hb(\nabla z,\nabla z)=0}{=}  -\frac{T_bC_1 \Cb_1 P\Pb}{(2\Xib -\Cb_3)\Ph^2}(\lb^a l^e -l^a \lb^e + i\varepsilon^{ae}{}_{\lb l})[\varepsilon_{aT\lb l}+E_a(z)][(l\cdot T)\lb_e -(\lb \cdot T) l_e + E_e(y)]\\
&=-\frac{T_bC_1 \Cb_1 P\Pb}{(2\Xib -\Cb_3)\Ph^2}(\lb^a l^e -l^a \lb^e + i\varepsilon^{ae}{}_{\lb l})E_a(z)E_e(y)-\frac{iT_bC_1 \Cb_1 P\Pb}{(2\Xib -\Cb_3)\Ph^2} \varepsilon^{ae}{}_{\lb l}\varepsilon_{aT\lb l}E_e(y)\\
&\quad+\frac{T_bC_1 \Cb_1 P\Pb}{(2\Xib -\Cb_3)\Ph^2}[(l\cdot T)\lb^a + (\lb \cdot T) \lb^a]E_a(z).
\end{align*}
We now turn to the first term $R(1)_b$, we show that its imaginary part cancels the term $I_1$ up to an error term:
\begin{align*}
 \Im(R(1)_b)&= - \Im(-2 \Cb_1 \Pb \nabla_a(C_1 P) \nabla_b (C_1 P) \nabla^a z)\\
 &= - \Im(-2 (y-iz) \nabla_a(y+iz) \nabla_b (y+iz) \nabla^a z)\\
&= - 2(y\nabla_b y  + z\nabla_b z)(\nabla z \cdot \nabla z)+2(z\nabla_b y-y\nabla_b z)(\nabla y \cdot \nabla z)\\
&= -I_1 + (z\nabla_b y-y\nabla_b z) \Im ((1-\frac{P^4}{\Ph^4})T^2).
\end{align*}
We now show that, up to an error term, $R(4)_b$ and $R(5)_b$ cancel out.
\begin{align*}
R(4)_b &= -\frac{4C_1^2 \Cb_1 P^3 \Pb}{2\Xib - \Cb_3}T^2 \H_{ae}\Hb_b{}^{e} \nabla^a z\\
       &= -\frac{4C_1^2 \Cb_1 P^3 \Pb}{2\Xib - \Cb_3}T^2\nabla^a z\frac{1}{2C_1 \Ph^2}(\lb_a l_e -l_a \lb_e + i\varepsilon_{ae\lb l})\frac{1}{2\Cb_1 \Phb^2}(\lb_b l^e -l_b \lb^e - i\varepsilon_{b}{}^e{}_{\lb l})\\
       &= -\frac{C_1 P^3 \Pb T^2}{(2\Xib - \Cb_3)\Ph^2 \Phb^2}  \nabla^a z (\lb_a l_b +l_a \lb_b + \varepsilon_{ae\lb l} \varepsilon_{b}{}^{e}{}_{\lb l})\\
       &= -\frac{C_1 P^3 \Pb T^2}{(2\Xib - \Cb_3)\Ph^2 \Phb^2}  \nabla^a z (2 \lb_a l_b +2l_a \lb_b + g_{ab})\\
       &= -\frac{C_1 P^3 \Pb T^2}{(2\Xib - \Cb_3)\Ph^2 \Phb^2}  \nabla_b z  -\frac{2C_1 P^3 \Pb T^2}{(2\Xib - \Cb_3)\Ph^2 \Phb^2} (\varepsilon^a{}_{T\lb l}+E^a(z))(\lb_a l_b +l_a \lb_b)\\
       &= -\frac{C_1 P^3 \Pb T^2}{(2\Xib - \Cb_3)\Ph^2 \Phb^2}  \nabla_b z  -\frac{2C_1 P^3 \Pb T^2}{(2\Xib - \Cb_3)\Ph^2 \Phb^2} [(E(z) \cdot \lb)l_b +(E(z) \cdot l)\lb_b].
\end{align*}
Because of the presence of $E(z)$, it is easy to see that the last term in the previous computations is an error term.
\begin{align*}
 R(5)_b &= \frac{C_1 P}{(2\Xib - \Cb_3)\Pb}(\nabla_a (C_1 P) \nabla^a z\nabla_b (\Cb_1\Pb) +\nabla_b (C_1 P) \nabla_a (\Cb_1 \Pb) \nabla^a z-\nabla (C_1 P) \cdot \nabla (\Cb_1\Pb) \nabla_b z)\\
&=  \frac{C_1 P}{(2\Xib - \Cb_3)\Pb}[2(\nabla z \cdot \nabla y) \nabla_b y +(\nabla z \cdot \nabla z -\nabla y \cdot \nabla y)\nabla_b z]\\
&=  \frac{C_1 P}{(2\Xib - \Cb_3)\Pb}[2 \nabla z \cdot \nabla y \nabla_b (C_1 P) + \frac{P^4}{\Ph^4}T^2 \nabla_b z]\\
&=  \frac{P^2 \Phb^2}{\Pb^2 \Ph^2} \frac{C_1 P^3 \Pb T^2}{(2\Xib - \Cb_3)\Ph^2 \Phb^2} \nabla_b z + \frac{C_1^2 P}{(2\Xib - \Cb_3)\Pb}\Im((1-\frac{P^4}{\Ph^4})T^2)\nabla_b P.
\end{align*}
Finally, we have
\begin{align*}
R(4)_b + R(5)_b &= (\frac{P^2 \Phb^2}{\Pb^2 \Ph^2}-1) \frac{C_1 P^3 \Pb T^2}{(2\Xib - \Cb_3)\Ph^2 \Phb^2} \nabla_b z + \frac{C_1^2 P}{(2\Xib - \Cb_3)\Pb}\Im((1-\frac{P^4}{\Ph^4})T^2)\nabla_b P\\
&\quad  -\frac{2C_1 P^3 \Pb T^2}{(2\Xib - \Cb_3)\Ph^2 \Phb^2} [(E(z) \cdot \lb)l_b +(E(z) \cdot l)\lb_b].
\end{align*}
Again, with the help of this formula, one shows easily that $R(4)_b + R(5)_b$ is an error term. In fact, we will show in the next section that $P$ and $\Ph$ are the same up to an error term by looking at the asymptotics. So we define the following error term
\begin{align}\label{E3}
 E(3)_b &=\Im(R(4)_b + R(5)_b).
\end{align}
Now we treat the last term. Notice that a certain amount of
calculation has already been done when we deal with the term
$R(4)_b$:
\begin{align*}
 R(6)_b &= -4C_1^2 \Cb_1 P^3 \Pb \H_{cb} [\H_a{}^c(2\Xib -\Cb_3)+\Hb_a{}^c(2\Xi -C_3)]\nabla^a z\\
&= C_1^2 \Cb_1 P^3 \Pb (2\Xib -\Cb_3) \H \cdot \H \nabla_b z + 4C_1^2 \Cb_1 P^3 \Pb \H_{bc} \Hb_a{}^c(2\Xi -C_3)\nabla^a z\\
&=  -\frac{\Cb_1 P^3 \Pb}{\Ph^4} (2\Xib -\Cb_3) \nabla_b z-\frac{(2\Xib -\Cb_3)(2\Xi -C_3)}{T^2}R(4)_b\\
&= -\frac{\Cb_1 P^3 \Pb}{\Ph^4} (2\Xib -\Cb_3) \nabla_b z +\frac{C_1 P^3 \Pb }{\Ph^2 \Phb^2}(2\Xi -C_3)  \nabla_b z +\frac{2C_1 P^3 \Pb }{\Ph^2 \Phb^2}(2\Xi -C_3) [(E(z) \cdot \lb)l_b +(E(z) \cdot l)\lb_b].
\end{align*}
so we have
\begin{align*}
 \Im(R(6)_b) &= -2 \Im(\frac{\Cb_1 \Pb}{P} (2\Xib -\Cb_3) \nabla_b z)  -\Im[\Cb_1(\frac{P^4}{\Ph^4}-1)(2\Xib -\Cb_3) -C_1 (\frac{P^2\Pb^2}{\Ph^2\Phb^2}-1)\frac{P}{\Pb}(2\Xi -C_3)]\nabla_b z \\
&\quad +\Im(\frac{2C_1 P^3 \Pb }{\Ph^2 \Phb^2}(2\Xi -C_3) [(E(z) \cdot \lb)l_b +(E(z) \cdot l)\lb_b])\\
&= -2 \Im(\frac{\Cb_1 \Pb}{P} (2\Xib -\Cb_3) \nabla_b z) + E(4)_b.
\end{align*}
where we have another error term:
\begin{align*}
 E(4)_b &=-\Im[\Cb_1(\frac{P^4}{\Ph^4}-1)(2\Xib -\Cb_3) -C_1 (\frac{P^2\Pb^2}{\Ph^2\Phb^2}-1)\frac{P}{\Pb}(2\Xi -C_3)]\nabla_b z \notag\\
&\qquad +\Im(\frac{2C_1 P^3 \Pb }{\Ph^2 \Phb^2}(2\Xi -C_3) [(E(z) \cdot \lb)l_b +(E(z) \cdot l)\lb_b])
\end{align*}
Finally, by putting all the identities together, we obtain
\begin{align*}
 \nabla_b U &= {I_1 + I_2 + I_3= I_1 + I_2} +\Im(2|C_1 P^2|\nabla^a z \nabla_a \nabla_b (C_1 P)) \\
&= 2z\nabla_b z -2\Im(T^2 \Cb_1 \Pb \nabla_b z)-2\Im(\frac{\Cb_1 \Pb}{P} (2\Xib -\Cb_3) \nabla_b z) +E(2)_b+E(3)_b+E(4)_b\\
&= 2|C_1 P|^2 \nabla^a z\Im(E(2))_b+E(3)_b+E(4)_b.
\end{align*}
This completes the proof of the lemma.
\end{proof}

\section{The study of the function $y$}
\subsection{The domain of definition for $y$}
In this subsection, we use a bootstrap argumen to prove that the electromagnetic potential $\Xi$ does not vanish on the hypersurface $\Sigma$. In particular, it shows that $y$ is a well-defined smooth function on $\Sigma$. The proof roughly goes as follows: We first use asymptotic information at spatial infinity bound $\nabla P$. Then, we can integrate this bound from spatial infinity to show that $P$ can not blow up, i.e. $\Xi$ has no zeroes. The study of the asymptotic behavior of various geometric quantities is the first step of the bootstrap argument. This is also the step where we show that all the algebraic terms in previous section are truly error terms (with bounds in terms of $\varepsilon$).

\subsubsection{Asymptotic identities}
In terms of the Cartesian coordinates mear spatial infinity (where we take $\partial_0 = T$), the assumption ($\bf A_1$) implies
\begin{equation*}
{ g^{00} }=-1-\frac{2M}{r}+O_4(r^{-2}),\quad  { g^{ij} } =\delta_{ij}-\frac{2M}{r}\delta_{ij}+O_4(r^{-2}),\quad  { g^{0i} }=-\varepsilon_{ijk}\frac{2 S^j x^k}{r^3}+O_4(r^{-3}),
\end{equation*}
and
\begin{equation*}
 { H_{0i} }=-\frac{q x_i}{r^3} + O_4(r^{-3}),\quad { H_{ij} }=\left(\frac{3 S_l x^l}{r^2}x^k-S^k\right)\frac{q\varepsilon_{ijk}}{Mr^3}+O_4(r^{-4}).
\end{equation*}
The Hodge dual $^*\!H_{ab}$ can be expressed as
\begin{align*}
 ^*\!H_{0i}&=\frac{1}{2} \varepsilon_{0iab}H_{cd}g^{ac} g^{bd}=\left(\frac{3 S_l x^l}{r^2}x^i-S^i\right)\frac{q}{Mr^3} + O_4(r^{-4}),\\
 ^*\!H_{ij}&=\frac{1}{2} \varepsilon_{ijab}H_{cd}g^{ac} g^{bd}=\varepsilon_{ijk}\frac{q x^k}{r^3}+O_3(r^{-3}).
\end{align*}
Therefore, we can compute the asymptotics of $\H^2$ where $\H_{ab} = \frac{1}{2}(H_{ab}+i{}{}^*\!H_{ab})$:
\begin{align*}
 \H^2 &= \H_{ab}\H^{ab}=\frac{1}{4}(H_{ab}H^{ab}-{}^*\!H_{ab}{}^*\!H^{ab})+\frac{i}{2}H_{ab}{}^*\!H^{ab}\\
      &= \frac{1}{2} H_{ab}H^{ab}+\frac{i}{2}H_{ab}{}^*\!H^{ab} \\
      &= -\frac{q^2}{r^4} + O(r^{-5}) + i\left(4\frac{q^2 S^l x_l}{M r^6} + O(r^{-6})\right).
\end{align*}
The Ernst 2-form is given as follows
\begin{equation*}
F_{0j}=2M \frac{x^j}{r^3} +O(r^{-3}), \quad F_{ij}=\frac{2}{r^3}\varepsilon_{ijk}S^k - \frac{6 S^l x_l}{r^5}\varepsilon_{ijk}x^k+O(r^{-4}).
\end{equation*}

We can fix a gauge of $\H$ at spatial infinity so that $q_B =0$. Hence, $C_1 = q = q_E + i q_B$. Since $\Ph^{-4} = -C_1^2 \H^2$, we have
\begin{equation*}
 \frac{1}{\Ph} = \frac{q}{r}+O(r^{-2})-i\left[\frac{q S^l x_l}{Mr^3} + O(r^{-3})\right].
\end{equation*}
Recall that $\nabla_b \Xi = \H_{ab}T^a =\H_{0b}$, thus,
\begin{equation*}
 \nabla_i \Xi =  -\frac{q x_i}{2r^3} + O_4(r^{-3})+ i\left[\frac{1}{2}(\frac{3 S_l x^l}{r^2}x^i-S^i)\frac{q}{Mr^3} + O_4(r^{-4})\right].
\end{equation*}
We impose the renormalization condition $\Xi = 0$ at infinity. We integrate the above equation to derive
\begin{equation*}
 \Xi = \frac{1}{2} \left[\frac{q}{r}+O(r^{-2})+i\left(-\frac{q S^l x_l}{Mr^3}+O(r^{-3})\right)\right],
\end{equation*}
which implies
\begin{equation*}
  \frac{1}{P} = \frac{q}{r}+O(r^{-2})-i\left[\frac{q S^l x_l}{Mr^3} + O(r^{-3})\right].
\end{equation*}
In particular, near infinity, $\frac{1}{\Ph}$ and $\frac{1}{P}$ agree up to first order in $r^{-1}$, i.e.
\begin{equation}
 |\frac{1}{\Ph} -\frac{1}{P}| = O_2(r^{-2}).
\end{equation}
We turn to the asymptotics for functions $y$ and $z$.
\begin{equation*}
 \frac{1}{C_1 P} = \frac{1}{r}+O(r^{-2})-i\left[\frac{ S^l x_l}{Mr^3} + O(r^{-3})\right].
\end{equation*}
Hence,
\begin{equation*}
 y = r + O(1), \quad  z = \frac{S^l x_l}{Mr} + O(r^{-1}).
\end{equation*}
We can prove an asymptotical version of Mars type lemma: The previous aysmptotics and the fact that $\nabla_0 z = \nabla_T z =0$ imply
\begin{align*}
 (\nabla z)^2 = \nabla_i z \nabla^i z &= \sum_{i=1,2,3}\left[\partial_i(\frac{S^l x_l}{Mr})\right]^2 +O(r^{-3})\\
 &= \frac{|S|^2}{M^2 r^2}-\frac{(S^l x_l)^2}{M^2 r^4}+O(r^{-3}).
\end{align*}
Hence,
\begin{align*}
 z^2 + (y^2+z^2)(\nabla z)^2 &= r^2\left(\frac{|S|^2}{M^2 r^2}-\frac{(S^l x_l)^2}{M^2 r^4}\right) + \frac{(S^l x_l)^2}{M^2 r^2} +O(r^{-1})\\
 &=\frac{|S|^2}{M^2} +O(r^{-1}).
\end{align*}
i.e.,
\begin{equation}\label{Mars_at_infinity}
 z^2 + (y^2+z^2)(\nabla z)^2 = \mathfrak{a}^2 + O(r^{-1}).
\end{equation}

\subsubsection{Bootstrap assumptions}
We study the exhaustion $\displaystyle \bigcup_{R\geqslant 0} E_R$ of $\Sigma$, where $E_R$'s are adapted to the bifurcate sphere $\S$:
\begin{equation*}
 E_{R} = \big\{p \in \Sigma~\big|~  d(p, \S) \leqslant R\big\},
\end{equation*}
where $d$ is the distance function on the space-like hypersurface $\Sigma$. For each $R$, $E_R$ is a connected set. In view of the asymptotical flatness, if $R$ is large enough (such that the function $r$ makes sense), on the boundary of $E_R$, $r$ is almost $R$ up to a constant.

We will run a bootstrap argument with the help of the asymptotical behavior of various geometric quantities. Since
\begin{equation*}
\Xi = \frac{q}{2r} + O(r^{-2}),
\end{equation*}
we can choose a $R^*$, such that on $\Sigma - B_{R^*}$ (the set $B_{R^*}$ is well defined when $R^*$ is sufficiently large), $\Xi$ is never zero. The number $R^*$ is fixed once forever and it will be useful later on. In particular, for $r\geqslant R^*$, we have
\[|P|\lesssim_{R^*} 1.\]

We then choose a large $R_0$, such that on $\Sigma-B_{R_0}$ (the set $B_{R_0}$ is well defined when $R_0$ is sufficiently large), $\Xi$ is not vanishing and on $\partial B_{R_0}$ we have
\begin{equation*}
 |2\Xi| \geqslant R_0^{-2}.
\end{equation*}
The larege radius $R_0$ will be determined later on.

Let
\[A=\big\{r\in [0,R_0]\big|\text{for all $p \in B_{R_0}-E_{r_0}$, we have} |2\Xi| \geqslant R_0^{-2}\big\}.\]
We define
\[r_0=\inf A.\]
In view of the asymptotics, we can take sufficiently large $R_0$ so that $r_0 \leqslant R_0 -1$.

We will use a bootstrap  argument to improve the previous inequality of $\Xi$, i.e., to show that there is a sufficiently large $R_0$ so that we indeed have
\[ |2\Xi| \geqslant 2 R_0^{-2}\]
for all $p \in B_{R_0}-E_{r_0}$. Once we have proved the above improved estimate, the standard continuity argument implies that $r_0 = 0$, i.e., $\Xi$ does not vanish. Hence, $P$ is well-defined.

\begin{center}
\includegraphics[width=2.8in]{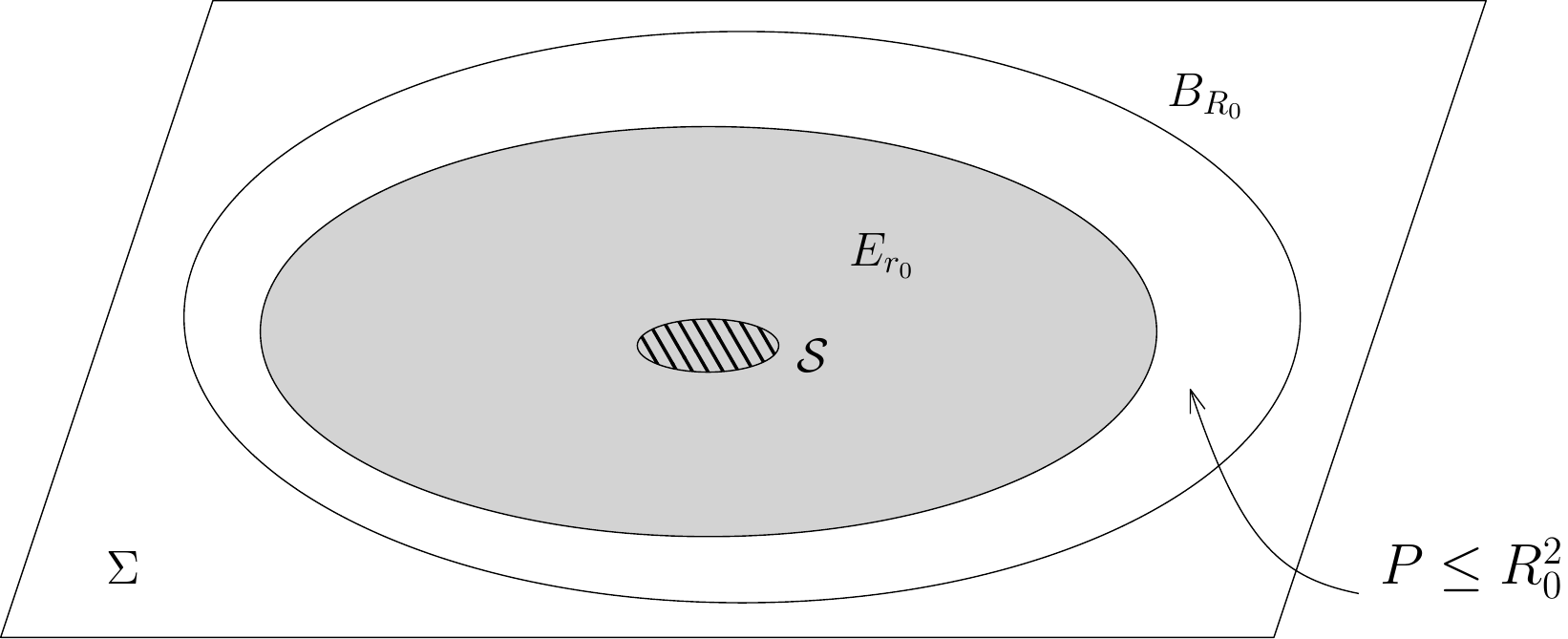}
\end{center}

In view of $P=(2\Xi)^{-1}$, it is suffices to prove the following lemma:
\begin{lemma}\label{bootslemma}
Under the above bootstrap assumptions, for all $p \in B_{R_0}-E_{r_0}$, we have
\begin{equation}
 |P| \leqslant \frac{1}{2} R_0^2
\end{equation}
\end{lemma}

\subsubsection{Approximate identities under bootstrap assumptions}
We will use the following auxiliary lemma repeatedly to integrate geometric quantities from spatial infinity to a finite region:
\begin{lemma}\label{ODE}
Let $\delta>0$ be a constant and $f$ be a smooth function defined on $\Sigma-E_{r_0}$. We assume that
\begin{itemize}
\item[1)] For all $x\in \Sigma-E_{r_0}$, we have
\begin{equation} \label{a}
|\nabla_b f| \lesssim \delta.
\end{equation}
\item[2)]For $r \rightarrow \infty$, i.e., on the region near spatial infinity, we have
\begin{equation}\label{b}
 |f| \lesssim r^{-\alpha},
\end{equation}
where $\alpha >0$ is a constant.
\end{itemize}
Then, we have
\begin{equation*}
 |f| \lesssim \min(\delta^{\frac{\alpha}{\alpha+1}},r^{-\alpha}).
\end{equation*}
\end{lemma}
\begin{proof}
 We fix a $p\in \Sigma-E_{r_0}$. Let $M = \delta^{-\frac{1}{\alpha+1}}$. If  $r(p) \geqslant M$,  the lemma automatically holds in view of \eqref{b}. If  $r(p) \leqslant M$, one can integrate \eqref{b} from a point $q$ with $r(q) = M$ to the point $p$. Thus, we have
\begin{align*}
  |f(p)| & \leqslant |f(q)|  +\|\nabla_b f\|_{L^\infty} R_0 \\
 &\lesssim \delta^{\frac{\alpha}{\alpha+1}} + \delta \delta^{-\frac{1}{\alpha+1}}.
\end{align*}
This completes the proof of the lemma.
\end{proof}
In the previous section, we derived various identities with algebraic error terms. Under the bootstrap assumptions, we now show that those algebraic error terms can be bounded in terms of $\varepsilon$. In view of Lemma \ref{lemmaPPh}, we fisrt derive
\begin{lemma}\label{bootsPPh}
 On $\Sigma-E_{r_0}$, we have
\begin{equation}\label{PminusPh}
 \Big|\frac{1}{P} - \frac{1}{\Ph}\Big| \lesssim \min(\varepsilon^{\frac{2}{3}}, r^{-2}).
\end{equation}
\end{lemma}
\begin{proof}
Lemma \ref{lemmaPPh} implies
\begin{align}\label{a1}
\nabla_b(\frac{1}{P}-\frac{1}{\Ph})&= \underbrace{\frac{C_1}{4}\Ph^3\H \cdot \frac{\nabla_b \B}{\Cb_3-2\Xib}}_{\mathbf{I}}-\frac{C^2_1}{2}\Ph^3 \frac{\Q_{dbef}}{\Cb_3-2\Xib}\H^{ef}T^d  \notag \\
&\qquad + \frac{3}{4C_1} \frac{\B_{db}}{\Cb_3-2\Xib}T^d-\frac{C_1}{4}\Ph^4 (\frac{\B}{\Cb_3-2\Xib}\cdot \H) \H_{db}T^d.
\end{align}
We start with the estimate the first term at the right hand side, i.e., $\mathbf{I}=\frac{C_1}{4}\Ph^3\H \cdot \frac{\nabla_b \B}{\Cb_3-2\Xib}$. Near spatial infinity, according to the asymptotics derived in the previous subsection, we have
\begin{equation*}
 |\Ph| \lesssim r, \ \ |\H| \lesssim r^{-2},\ \ |\nabla_b \B| \lesssim r^{-3}.
\end{equation*}
Therefore, near spatial infinity, we have
\begin{equation*}
  \mathbf{I}\lesssim \frac{1}{r^2}.
\end{equation*}
On the other hand, $\Ph^3\H$ is bounded. Therefore, in view of the assumption $\bf{A_3}$, we have
\begin{equation*}
  \mathbf{I}\lesssim \varepsilon.
\end{equation*}
To summarize, we have
\begin{equation*}
 \mathbf{I} \lesssim \min(\varepsilon, \frac{1}{r^2}).
\end{equation*}
We can proceed in the same manner to bound other terms at the right hand side of \eqref{a1}. This leads to
\begin{equation*}
  \Big|\nabla_b(\frac{1}{P}-\frac{1}{\Ph})\Big| \lesssim \min(\varepsilon, \frac{1}{r^2}).
\end{equation*}
We can apply the Lemma \ref{ODE} with $\delta = \varepsilon$ and $\alpha =2$ to obtain \eqref{PminusPh}.
\end{proof}
In principle, one can bound all the algebraic error terms by $\varepsilon$ and $r^{-1}$ in a similar way. While in the sequel, we shall estimate terms which are necessary in the proof.
\begin{lemma}\label{booxy}
On $\Sigma-E_{r_0}$, we have the following approximate wave equation for $y$
\begin{equation}\label{booooxy}
\Big|\Box_g y + \frac{2(M-y)}{y^2+z^2}\Big| \lesssim \min(\varepsilon^{\frac{1}{3}}, r^{-2}),
\end{equation}
provided $R_0^3\varepsilon^{\frac{1}{3}}\leqslant 1$.
\end{lemma}
\begin{proof}
According to Corollary \ref{cornablas}, it suffices to control the following algebraic error
\begin{equation*}
 e(y) = \Re[-\underbrace{(\frac{1}{P^4}-\frac{1}{\Ph^4})\frac{2P^3}{C_1}(\frac{C_2}{\Pb}-1)}_{\mathbf{I}_1}+ \underbrace{P^2 \H \cdot \B}_{\mathbf{I}_2} -\underbrace{4C_1P^3 \H^2 \Re(V)]}_{\mathbf{I}_3}.
\end{equation*}
First of all, thanks to the asymptotics near spatial infinity, it is straightforward to check that
\[|e(y)|\lesssim r^{-2}.\]
Secondly, to obtain the bound in terms of $\varepsilon$, we bound the terms in $e(y)$ one by one. For $\mathbf{I}_1$, we have
\begin{align*}
|\mathbf{I}_1|\lesssim \Big|\frac{1}{P}-\frac{1}{\Ph}\Big| \left(\frac{1}{|P|}+\frac{1}{|\Ph|}\right)\left(\frac{1}{|P|^2}+\frac{1}{|\Ph|^2}\right)\cdot|P|^3\big(\frac{1}{|\Pb|}+1\big)
\end{align*}
Since $\Sigma$ is well-defined and for $x\in \Sigma-B_{R^*}$ we have $|P|^{-1}\lesssim 1$ (this also implies that $|\Ph|^{-1}\lesssim 1$ because of \eqref{PminusPh}). Therefore, $|P|^{-1}$ and $|\Ph|^{-1}$ are a priori bounded on $\Sigma$. Therefore, under that bootstrap assumption, we can use Lemma \ref{bootsPPh} to obtain
\begin{align*}
|\mathbf{I}_1|\lesssim \varepsilon^{\frac{2}{3}}R_0^{3}\lesssim \varepsilon^{\frac{1}{3}},
\end{align*}
provided $R_0^3\varepsilon^{\frac{1}{3}}\leqslant 1$.

To bound $\mathbf{I}_2$, we simply use the assumption $\bf{A_3}$. To bound $\mathbf{I}_3$, we need to bound $V$: we repeat the process  by using Lemma \ref{ODE} and \eqref{potentialV}. This completes the proof.
\end{proof}
Finally, we need the following version of Mars type lemma which will be the key for the rest of the chapter
\begin{lemma}\label{bootstrapmars}
 On $\Sigma -E_{r_0}$, we have
\begin{align}
 \begin{cases}
  (\nabla z)^2 &= \frac{1}{y^2 +z^2}(\mathfrak{a}^2-z^2) + O(\varepsilon^{\frac{1}{3}}),\\
  (\nabla y)^2 &= \frac{1}{y^2 +z^2}(y^2 -2M y + q^2 + \mathfrak{a}^2) + O(\varepsilon^{\frac{1}{3}}),
 \end{cases}
\end{align}
provided $R_0^3\varepsilon^{\frac{1}{3}}\leqslant 1$.
\end{lemma}
\begin{proof}
We start with the first estimate which relies on Lemma \ref{algebraicmars}. We need to bound all the algebraic error
terms in the proof of Lemma \ref{algebraicmars}: first of all, one encounters a term involving $\B_{ab}$ and $\Q_{abcd}$, one uses the assumption $\bf{A_3}$; secondly, when one encounters a term
involving $P- \Ph$ or $\frac{1}{P}- \frac{1}{\Ph}$, one uses Lemma \ref{bootsPPh}. By repeatedly using these two principles, we can derive correct estimates for all the algebraic errors.  Finally, one can use Lemma \ref{ODE} to conclude. In view of Corollary \ref{cornablas}, the second identity follows from the first one.
\end{proof}

\subsubsection{Completion of the bootstrap argument}
We divide $B_{R_0}-E_{r_0}$ into two parts $B_{R_0}-E_{r_0} = W_1 \cup W_2$ where
\begin{align*}
 W_1&=\big\{p \in B_{R_0}-E_{r_0} \big| |P| \leqslant \frac{1}{2} R_0^2 \big\}, \\
 W_2&=\big\{p \in B_{R_0}-E_{r_0} \big| \frac{1}{2} R_0^2 \leqslant |P| \leqslant  R_0^2 \big\}.
\end{align*}
Since the Lemma \ref{bootslemma} automatically holds on $W_1$, it suffices to show that, on $W_2$, $|\nabla_b P| \leqslant C$ where the
constant $C$ is independent of $R_0$ and $r_0$. In view of the fact
\begin{equation*}
 |\frac{1}{P}| \leqslant \frac{2}{R_0^2},
\end{equation*}
we know that $|\frac{1}{P}| $ is small since $R_0$ is chosen to be large.
Recall that
\begin{equation*}
T^2 = g(T,T) = -|\frac{1}{P}-C_2|^2 + |C_2|^2 - 1-2\Re(V).
\end{equation*}
We have already shown that the algebraic error term $\Re(V) \lesssim \varepsilon^{\frac{1}{3}}$, so if $R_0$ is sufficiently large and $\varepsilon$ is sufficiently small, we have
\begin{equation*}
 g(T,T) \in [-\frac{101}{100}, -\frac{99}{100}].
\end{equation*}
In particular, this shows $T^a$ is time-like on $W_2$. Since $\nabla_T y = \nabla_T z =0$, one knows $\nabla y$ and $\nabla z$ are space-like on $W_2$. In view of Lemma \ref{bootstrapmars} and the fact that
$z^2 + y^2 \sim |P|^2 \geqslant \frac{1}{2}R_0^2$ on $W_2$, we derive
\begin{equation*}
 |\nabla_a z \nabla^a z| + |\nabla_a y \nabla^a y| \leqslant C,
\end{equation*}
where $C$ is independent of $R_0$. If we let $e_{(1)},e_{(2)},e_{(3)}$ be orthogonal to $T$ and be an orthonormal space-like basis, we then derive the following derivative bound
\begin{equation}
 |\nabla_{e(i)} P| \leqslant C.
\end{equation}
Thus, we integrate $\nabla P$ along a geodesics connecting $p_{1}$ and $p_{2}$ on $\Sigma$, where $p_{1} \in B_{R_0} - E_{r_0}$  and $p_{2}$ is on the boundary of $B_{R_0}$, we have
\begin{equation*}
 |P(p_{1})| \leqslant |P(p_{2})| + C R_0
\end{equation*}
Since $r(p_{2})$ is close to $R_0$, the asymptotics of $P$ yields
\begin{equation*}
 |P(p)| \leqslant \frac{M}{R_0} + C R_0.
\end{equation*}
Once we choose a large $R_0$, this implies
\begin{equation*}
 |P| \leqslant \frac{1}{2} R_0^2.
\end{equation*}
This completes the proof of Lemma \ref{bootslemma} hence closes the bootstrap argument. In particular, all the estimates in Lemma \ref{booxy} and Lemma \ref{bootstrapmars} are valid on the hypersurface $\Sigma$.

\subsection{The analysis of the function $y$ near horizons}
In this section, we study the behavior of $y$ near the bifurcate sphere. We first show that $y$ is almost a constant on $\S$. We also show that when one moves away from $\S$ towards infinity, in a small
neighborhood of $\S$, the function $y$ is increasing. We shall use the standard double null foliation with respect to the bifurcate sphere $\S$ defined in Section \ref{double_null_foliation}. In particular, we will use the null pair $(L, \Lb)$ and the optical functions $u$ and $\ub$. Recall that $L$ and $\Lb$ are the null generators of the
future and past horizons $\H^+$ and $\H^-$. We begin with the following lemma:
\begin{lemma}
On the bifurcate sphere $\S$, we have
\begin{equation}
 \big|y-(M+\sqrt{M^2-q^2-\mathfrak{a}^2})\big| \lesssim \varepsilon^{\frac{1}{3}}.
\end{equation}
\end{lemma}

\begin{remark} The quadratic polynomial $y^2 - 2My + q^2 +\mathfrak{a}^2$ plays a crucial role in our analysis in the sequel. We use $y_+$ and $y_-$ to
denote its two roots, i.e.,
\begin{equation*}
 y_\pm = M \pm \sqrt{M^2-q^2-\mathfrak{a}^2}.
\end{equation*}
\end{remark}

\begin{proof}
By virtue of the non-expanding properties for the stationary horizons, it is easy to derive
\begin{equation*}
 \F_{\alpha \beta}L^\beta = \F_{L\Lb} L_\alpha, \qquad  \F_{\alpha \beta}\Lb^\beta = -\F_{L\Lb} \Lb_\alpha
\end{equation*}
where $\F$ is the anti-self-dual complexification of the Ernst 2-form $F_{ab}$.
Recall that $l$ and $\lb$ are also defined as the principal null directions (eigenvectors)of $\H_{ab}$, i.e.,  $\H_{ab} l^b$ is proportional to $l$ and also is proportional to $\lb$. According to the definition of the tensor $\B_{ab}$, the smallness
assumption of $\B_{ab}$ should be understood as an almost alignment
condition between $\F_{ab}$ and $\H_{ab}$. When $\varepsilon$ is small
enough, up to a conformal rescaling and a relabeling, one assume that, on the bifurcate sphere $\S$, we have
\begin{equation}
 |L - l| \lesssim \varepsilon ,\ \ |\Lb - \lb| \lesssim \varepsilon.
\end{equation}
Since $T$ is tangent to $\S$, so $g(L,T) = g(\Lb,T)=0$. In view of the null decompsition of $\nabla y$
\begin{equation*}
\nabla_b y = (l\cdot T)\lb_b -(\lb \cdot T) l_b + E_b(y)
\end{equation*}
and the fact that the algebraic error term $E_b(y)$ can be
controlled by $\varepsilon^{\frac{1}{3}}$, thus on the bifurcate sphere $S_1$ we have
\begin{equation}
 |\nabla y| \lesssim \varepsilon^{\frac{1}{3}}.
\end{equation}
According to Lemma \ref{bootstrapmars}
\begin{equation*}
(\nabla y)^2 = \frac{1}{y^2 +z^2}(y^2 -2M y + q^2 + \mathfrak{a}^2) + O(\varepsilon^{\frac{1}{3}}).
\end{equation*}
So the smallness of $\nabla y$ implies that on $\S$ we have
\begin{equation*}
 |y^2 -2M y + q^2 + \mathfrak{a}^2| \lesssim  O(\varepsilon^{\frac{1}{3}}).
\end{equation*}
Thus on $\S$, we have two possibilities for $y$: it is either close
to the value $y_+$ or close to the value $y_-$.

We now eliminate the second alternative. We start by assuming the opposite that $|y - y_-| \lesssim \varepsilon^{\frac{1}{3}}$ and we will derive a contradiction.
The key is to use the wave equation \eqref{booooxy}
\begin{equation*}
\Box_g y = - \frac{2(M-y)}{y^2+z^2} + O(\varepsilon^{\frac{1}{3}}),
\end{equation*}
near the horizon. Since
$y_- = M - \sqrt{M^2-q^2-\mathfrak{a}^2}< M$, we have
\begin{equation}
 \Box_g y < 0
\end{equation}
provided $\varepsilon$ is small enough. We now fix a point $p \in \S$
and take a local coordinate $(w_1,w_2, w_3 = u, w_4 = \ub)$ around $p$ where $(w_1,w_2)$ is a local chart on $\S$. Since $|y-y_-| + |\nabla y|=O(\varepsilon^{\frac{1}{3}})$ on $\S$, by taking the Taylor expansion of $y$ with respect to the given coordinates system around $p$, we find that there is a smooth function $f$ so that
\[y = y_- + u \ub f + O(\varepsilon^{\frac{1}{3}}).\]
This implies
\[ -f =\Box_g y +O(\varepsilon^{\frac{2}{3}})= - \frac{2(M-y)}{y^2+z^2} + O(\varepsilon^{\frac{1}{3}}).\]
Hence $f >0$ on $\S$ provided $\varepsilon$ is sufficiently small. Since $\S$ is compact, there is a constant $\delta >0$, such that $f> \delta$. The key fact is, once $\varepsilon$ is small enough, the choice of $\delta$ is independent of $\varepsilon$. Since the domain of outer communication can be identified by $u \geqslant 0$ and $\ub \leqslant 0$. Thus on $\Sigma$, we have $u\ub \leqslant 0$. In fact, by shrinking the size of $u$, $\ub$ and $\varepsilon$, we have
\[y \leqslant y_- - \frac{1}{2}\delta |u \ub|.\]
This shows that if one is locally away from $\S$ along $\Sigma$, the function value of $y$ is strictly below the maximum of $y$ on $\S$.  Since $y \rightarrow +\infty$ when one moves towards spatial infinity on $\Sigma$, so
there must exist a point $p_0 \in \Sigma$, such that $y$ attains its minimum on $p_0$. Obviously, there is a universal $\delta'$, such that $y(p_0) < y_- -\delta'$. Now we take an orthonormal frame $e_0,e_1,e_2,e_3$ near $p_0$, such that $e_1,e_2,e_3$ is tangential to $\Sigma$. This forces $e_0$ to be time-like since $\Sigma$ is space-like. Thus $e_1(y) = e_2(y) = e_3(y) =0$ (since $p_0$ is a critical point of $y$ on $\Sigma$). Therefore, at the point $p_0$, we have
\[ \nabla^a y\nabla_a y = -e_0(y)^2 \leqslant 0.\]
We apply Lemma \ref{bootstrapmars} for the second time to obtain
\[ y^2 -2M y + (q^2 + \mathfrak{a}^2) \leqslant O(\varepsilon^{\frac{1}{3}}).\]
Hence,
\[y_--O(\varepsilon^{\frac{2}{3}}) \leqslant y(p_0) \leqslant y_+ +O(\varepsilon^{\frac{2}{3}}).\]
It contradicts the fact that $y(p_0) < y_- -\delta'$. This completes the proof.
\end{proof}

A similar argument shows that around $\S$, the function $y$ is increasing when one moves away from $\S$ towards infinity:
\begin{lemma}
  There is a small constant $\delta_1$ and a sphere $\tilde{\S} \in \Sigma$ which is close to $\S$ and can be deformed from $\S$ along the norm of $\S$ on $\Sigma$, such that for each $p \in \tilde{\S}$,
\begin{equation*}
 y(p) > y_+ + \delta_1 > \max_{p\in \S} y(p).
\end{equation*}
\end{lemma}
\begin{proof}
As in the proof of the previous lemma, we assume $y= y_+ + u\ub f+O(\varepsilon^{\frac{1}{3}})$ around the bifurcate sphere $\S$. Similarly, since on $\S$ we have $y=y_+ > M$, we deduce that
\[
 f = -\Box y + O(\varepsilon^{\frac{1}{3}}) = -\frac{2(y-M)}{x^2+y^2}<0.
\]
Similarly, a compactness consideration shows there is a constant $\delta > 0$ which is independent of $\varepsilon$, such that on $\S$,
\begin{equation*}
 f < - \delta.
\end{equation*}
Once again, near $\S$, $u\ub \leqslant 0$ and $y$ can be written as
\begin{equation*}
 y = y_+ + u\ub f + O(\varepsilon^{\frac{2}{3}}) \geqslant y+ \delta u \ub + O(\varepsilon^{\frac{2}{3}})
\end{equation*}
We consider the surface define by $\{x \in \Sigma| u(x)\ub(x) =-\varepsilon_1 \}$. It is clear that when $\varepsilon_1$ is sufficiently small but relatively much greater than $\varepsilon$, such a surface can be chosen as $\tilde{\S}$.
\end{proof}

\subsection{$T$-Conditional Pseudo-Convexity for $y$}
In this subsection, we prove a $T$-conditional pseudo-convexity	property (namely the inequality \eqref{T_P_convexity}) for the function $y$ away from the bifurcate horzions $\H^+ \cap \H^-$. This allows us to use Ionescu-Klainerman's Carleman estimates to conclude uniqueness for certain $T$-symmetric wave equations. We refer the readers to \cite{Ionescu_Klainerman_Kerr} and \cite{Alexakis_Ionescu_Klainerman} for details.
\begin{lemma}\label{T_pseudo_convexity}
 Assume $p \in \Sigma$ is a given point such that $y(p) > y_+ + \varepsilon^2$ where $\varepsilon$ is a small number as before. Then, there are two positive constants $c_1$ and $\mu$, so that $c_1$ is independent of $p$, $\mu$ may depend on $p$ and
\begin{equation}\label{T_P_convexity}
 X^a X^b \big(\mu g_{ab}(p)- \nabla_a \nabla_b y(p)\big) \geqslant c_1 |X|^2,
\end{equation}
for all vectors $X\in T_p \M$ with the property that
\begin{equation}
|X^a T_a (p)| + |X^a\nabla_a y(p)| \leqslant c_1|X|.
\end{equation}
\end{lemma}
\begin{proof} It is straightforward to check the estimates near spatial infinity. Therefore, it suffices to assume that $y$ and $z$ are bounded by a universal constant which is independent of $\varepsilon$. We also assume that $|X^a T_a (p)|+ |X^a\nabla_a y(p)| \leqslant c_1|X|$ where the constant $c_2$ will be determined later on. The key step is  the computation of the Hessian of $y$ in $X$ direction. In view of the definition of $y$, it suffices to compute $X^a X^b \nabla_a \nabla_b (C_1 P)$ and then take the real part:
\begin{align*}
 X^a X^b \nabla_a \nabla_b (C_1 P) & = \underbrace{-\frac{C_1}{P}( \nabla_X P)^2}_{I_1}  \underbrace{- \frac{T^2}{C_1 P}g(X,X)}_{I_2} \underbrace{-\frac{2C_1 P^2}{2\Xib - \Cb_3}T^2 \H_{ae}\Hb_b{}^e X^a X^b}_{I_3}\\
&\quad+ \underbrace{\frac{C_1}{(4\Xib - 2\Cb_3)\Pb^2}(2\nabla_X P \nabla_X \Pb - (\nabla P \cdot \nabla \Pb) g(X,X))}_{I_4}\\
&\quad \underbrace{-2C_1 P^2 \H_{cb}\big[\H_a{}^c(2\Xib -\Cb_3) + \Hb_a{}^c(2\Xi - C_3)\big]X^a X^b}_{I_5}.
\end{align*}
We turn to the bounds on the $I_i$'s. In the rest of proof, we shall use null frames to decompose vectors and tensors.

For $I_1$, since $|X^a\nabla_a y(p)| \leqslant c_1|X|$, we have
\begin{equation*}
 \Re(I_1) = \frac{y}{y^2 + z^2}(Xz)^2+O\left((c_1)^2\right)|X|^2.
\end{equation*}

For $I_2$, we have
\begin{equation*}
 \Re(I_2) = -\frac{y T^2}{y^2 + z^2}(X_1^2 +X_2^2)+\frac{y T^2}{y^2 + z^2}2X_3X_4.
\end{equation*}

For $I_3$, we have
\begin{align*}
 I_3 &= -\frac{2 C_1 P^2}{2\Xib -\Cb_3} \frac{T^2}{2C_1 P^2}(\lb_a l_e -l_a \lb_e + i\varepsilon_{ae\lb l})\frac{1}{2\Cb_1 \Pb^2}(\lb_b l^e -l_b \lb^e - i\varepsilon_{a}{}^{e}{}_{\lb l})X^aX^b\\
&=-\frac{T^2}{(2\Xib - \Cb_3)2\Cb_1 \Pb^2}(2\lb_a l_b + 2\lb_b l_a + g_{ab})X^aX^b\\
&=-\frac{T^2}{2}\frac{1}{(1-\Cb_3 \Pb)\Cb_1 \Pb}(X_1^2 + X_2^2 + 2X_3X_4)
\end{align*}
If we define
\begin{equation*}
 A = \Re(\frac{1}{(1-\Cb_3 \Pb)\Cb_1 \Pb}) = -\frac{q^2\big[(My-q^2)y-Mz^2\big]}{\big((My-q^2)^2 + M^2 z^2\big)(y^2 +z^2)},
\end{equation*}
we obtain
\begin{equation*}
 \Re(I_3) = -\frac{1}{2}A T^2 (X_1^2 + X_2^2)-\frac{1}{2}A T^2 2X_3 X_4.
\end{equation*}

For $I_4$, we have
\begin{align*}
 I_4 &= \frac{C_1 \Cb_1}{2(2\Xib-\Cb_3)\Pb^2 \Cb_1}\big(2\nabla_X P \nabla_X \Pb-(\nabla P \cdot \nabla \Pb)g(X,X)\big)\\
&=\frac{1}{2(1-\Cb_3 \Pb)\Cb_1\Pb}\big(2X(z)^2+O\left((c_1)^2\right)|X|^2-((\nabla y)^2 + (\nabla z)^2)(-2 X_3 X_4 +X_1^2 +X_2^2)\big),
\end{align*}
i.e.,
\begin{equation*}
 \Re(I_4) = A X(z)^2 - \frac{1}{2}A\big[(\nabla y)^2 + (\nabla z)^2\big](X_1^2 +X_2^2) + \frac{1}{2}A\big[(\nabla y)^2 + (\nabla z)^2\big]2X_3X_4+A\cdot O\left((c_1)^2\right)|X|^2.
\end{equation*}

For $I_5$, we have
\begin{align*}
 I_5 &= 2C_1 P^2 \H_{bc}\big[\H_{a}{}^c (2\Xib -\Cb_3) + \Hb_{a}{}^c (2\Xi -C_3)\big]X^aX^b\\
&= 2C_1 P^2 \frac{1}{4}\H \cdot \H g(X,X)(2\Xib - \Cb_3) + 2C_1 P^2 (2\Xi-C_3)X^a X^b\frac{1}{2C_1 P^2}\frac{1}{2\Cb_1\Pb^2}(2\lb_a l_b +2\lb_b l_a + g_{ab})\\
&= -\frac{2\Xib - \Cb_3}{2C_1 P^2}(-2 X_3 X_4 +X_1^2 +X_2^2) + \frac{2\Xi - C_3}{2\Cb_1 \Pb^2}(2 X_3 X_4 +X_1^2 +X_2^2)
\end{align*}
Let
\begin{equation*}
 B = \Re(\frac{2\Xib - \Cb_3}{C_1 P^2}) = -\frac{(My-q^2)y-Mz^2}{(y^2+z^2)^2}.
\end{equation*}
Therefore,
\begin{equation*}
 \Re(I_5) = 2 B  X_3 X_4.
\end{equation*}
Combining all the computations, we can derive
\begin{align*}
 E(X,X) &= \mu g(X,X) - X^a X^b \nabla_a \nabla_b y\\
&= -\left(\frac{y}{y^2 +z^2} + A\right)X(z)^2  + \left(\mu + \frac{y T^2}{y^2 +z^2} + A (\nabla z)^2\right)(X_1^2 +X_2^2)  \\
&\ \ \ - \left(\mu + \frac{y T^2}{y^2 +z^2}
+ A (\nabla y)^2 +B\right)2X_3 X_4\\
&= \underbrace{\left[-\frac{y}{y^2+z^2}X(z)^2 + \left(\mu + \frac{y T^2}{y^2 +z^2}\right)(X_1^2 +X_2^2)-\left(\mu + \frac{y T^2}{y^2 +z^2}+B\right)2X_3 X_4\right]}_{I}\\
&\quad + \underbrace{A\left[(\nabla z)^2(X_1^2 +X_2^2)-(\nabla y)^2 2X_3 X_4 -X(z)^2\right]}_{II}+(1+A)O\left((c_1)^2\right)|X|^2.
\end{align*}
\begin{lemma}
 There second term $II$ can be bounded as follows:
 \[|II|=A\cdot O\left((c_1)^2\right)|X|^2+O(\varepsilon^{\frac{2}{3}})\cdot A.\]
\end{lemma}
\begin{proof}
Recall that we have
\[ \nabla_b y =g(T,l)\lb_b - g(T,\lb)l_b+O(\varepsilon^{\frac{1}{3}}), \ \ \nabla_b z = \varepsilon_{bacd}T^a \lb^c l^d +O(\varepsilon^{\frac{1}{3}}),\]
and
\[T_c = -g(T,l)\lb_c -g(T,\lb)l_c - \varepsilon_{cbad}\lb^b l^d\nabla^a z.\]
On the other hand, the condition $|g(\nabla y,X)| \leqslant c_1|X|$ implies
\begin{equation}\label{gtlbX3}
 g(T,\lb)X_3 =g(T,l)X_4+O(c_1)|X|+O(\varepsilon^{\frac{1}{3}})|X|.
\end{equation}
The condition $|g(T,X)| \leqslant c_1|X|$ implies
\begin{equation*}
 g(T,\lb)X_3 + g(T,l)X_4 = X_1 \nabla_2 z-X_2 \nabla_1 z+O(c_1)|X|+O(\varepsilon^{\frac{1}{3}})|X|.
\end{equation*}
Therefore, we have
\begin{align*}
&\ \ \ \ (\nabla z)^2(X_1^2 +X_2^2)-(\nabla y)^2 2X_3 X_4 -X(z)^2\\
&=((\nabla_1 z)^2 + (\nabla_2 z)^2)(X_1^2 +X_2^2) - (X_1 \nabla_1 z +X_2 \nabla_2 z)^2 -(\nabla y)^2 2X_3 X_4\\
&=(X_1 \nabla_2 z -X_2 \nabla_1 z)^2 -(2g(T,\lb)g(T,l) 2X_3 X_4)+O\left((c_1)^2\right)|X|^2+O(\varepsilon^{\frac{2}{3}})\\
&=(g(T,\lb)X_3 + g(T,l)X_4)^2 - 4g(T,\lb)g(T,l) X_3 X_4+O\left((c_1)^2\right)|X|^2+O(\varepsilon^{\frac{2}{3}})|X|^2\\
&=O\left((c_1)^2\right)|X|^2+O(\varepsilon^{\frac{2}{3}})|X|^2.
\end{align*}
Therefore, we obtain
\[|II|=A\cdot O\left((c_1)^2\right)|X|^2+A\cdot O(\varepsilon^{\frac{2}{3}})|X|^2.\]
This completes the proof of the lemma.
\end{proof}

Therefore, we have
\begin{align*}
 E(X,X)
&= \underbrace{\left[-\frac{y}{y^2+z^2}X(z)^2 + \left(\mu + \frac{y T^2}{y^2 +z^2}\right)(X_1^2 +X_2^2)-\left(\mu + \frac{y T^2}{y^2 +z^2}+B\right)2X_3 X_4\right]}_{I}\\
&\quad +(1+A)O\left((c_1)^2\right)|X|^2+A\cdot O(\varepsilon^{\frac{2}{3}})|X|^2.
\end{align*}
In view of the Cauchy-Schwarz inequality, i.e., $|X(z)|^2 \leqslant (X_1^2+X_2^2)(\nabla z)^2$, we have
\begin{align*}
 I &\geqslant -\frac{y}{y^2+z^2}(\nabla z)^2(X_1^2 +X_2^2) + (\mu + \frac{y T^2}{y^2 +z^2})(X_1^2 +X_2^2)-(\mu + \frac{y T^2}{y^2 +z^2}+B)2X_3 X_4\\
&=  (\mu + \frac{y T^2}{y^2 +z^2}-\frac{y}{y^2+z^2}(\nabla z)^2)(X_1^2 +X_2^2)-(\mu + \frac{y T^2}{y^2 +z^2}+B)2X_3 X_4
\end{align*}

The first ingredient to bound $I$ is the followin ienquality:
\begin{equation}\label{zz1}
 -B >\frac{y}{y^2+z^2}(\nabla z)^2.
\end{equation}
This follows directly from the expression of  $B$,  the formula $\nabla z$ in \eqref{bootstrapmars}, the fact that $y > y_+ > M$ as well as non-extremality $\a^2 + q^2 <M^2$.

The second ingredient is the following claim: There is a constant $c_3 >0$, such that
\begin{equation}\label{zz2}
 2 X_3 X_4 \geqslant c_3 (X_3^2 +X_4^2)+O\left((c_1)^2\right)|X|^2+O(\varepsilon^{\frac{1}{3}})|X|^2.
\end{equation}
Once we show this inequality, we complete the proof. To see this, we use \eqref{gtlbX3} to derive
\[ 2 X_3 X_4 = \frac{g(T,\lb)}{g(T,l)}X_3^2 + \frac{g(T,l)}{g(T,\lb)}X_4^2+O(c_1)|X|^2+O(\varepsilon^{\frac{1}{3}})|X|^2.\]
According to Mars type lemma, we now that
\begin{align*}
2g(T,\lb)g(T,l) &= (\nabla y)^2+O(\varepsilon^{\frac{2}{3}}) \\
 &=\frac{1}{y^2+z^2}(y^2 -2My + q^2 +\a^2)+O(\varepsilon^\frac{1}{3})>0.
\end{align*}
In the above inequality, we used the fact that $y$ and $z$ are bounded and $\varepsilon$ is sufficiently small. Hence, here is a constant $c_3$ (depending on the particular point $p$) such that $\frac{g(T,\lb)}{g(T,l)}>c_3 $ and $\frac{g(T,l)}{g(T,\lb)}>c_3$. This proves \eqref{zz2}. If $c_1$ and $\varepsilon$ are sufficiently small, this indeed implies
\begin{equation}\label{zz3}
 2 X_3 X_4 \geqslant c_4 (X_3^2 +X_4^2),
\end{equation}
where $c_4>0$.

We then use \eqref{zz1} and \eqref{zz3} to bound $I$ and this leads to
\begin{align*}
 I \geqslant \geqslant c_2 (X_1^2 +X_2^2) + 2 c_2 X_3 X_4.
\end{align*}
Hence, if $c_1$ and $\varepsilon$ are sufficiently small, we can choose $\mu$ to obtain
\begin{align*}
 E(X,X) &\geqslant I(1+A)O\left((c_1)^2\right)|X|^2+A\cdot O(\varepsilon^{\frac{2}{3}})|X|^2\\
 &\geqslant c_5 (X_1^2 +X_2^2) + 2 c_5 X_3 X_4.
\end{align*}
where $c_5>0$. This completes the proof of the $T$-conditional pseudo-convexity.
\end{proof}

\section{The extension of Hawking vector field}
We now use the $T$-pseudo convexity  to construct a second Killing vector field $K$, namely the Hawking vector field, on the entire domain of outer communication $\E$. The vector field is constructed inductively along the level surfaces of the function $y$. To start the induction argument, we prove a slightly modified version of Theorem \ref{local_rigidity_1} and Theorem \ref{local_rigidity_2}:
\begin{proposition}\label{local}
There exists a neighborhood $\O$ of the bifurcate sphere $\S$
and a non-trivial Killing vector field $K$ in $\O$ which is tangent
to the null generators of $\H^+$ and $\H^-$ such that $\L_K H =0$
and $K(y)=0$. In addition, there is a constant $\lambda_0 \in
\mathbb{R}$, such that the vector field $Z = T +\lambda_0 K$ has
complete periodic orbits in $\O$.
\end{proposition}
\begin{proof}
 We only need to show $K(y)=0$ which does not appear in Theorem \ref{local_rigidity_1}. In fact, we will show that $K(\Xi) =0$ which, in view of the definition of $y$, implies $K(y)=0$. The proof is a routine calculation. We first notice that
 \[K(\Xi) = K^b \Xi_b = \H_{ab}K^b T^a = \H(T,K).\]
 We can assume $T \neq 0$ and
$K\neq 0$ around a point (otherwise the statement holds trivially on this point). Since $[T,K]=0$, according to Frobenius theorem, one can construct a local coordinate $(x^0, x^1,x^2,x^3)$, such that $\partial_0 =T$ and $\partial_1=K$. In view of the Maxwell equations, i.e., $d\H =0$, we have
\begin{align*}
0&\ \ \ =d\H(T,K,\partial_i)=d\H(\partial_0,\partial_1,\partial_i)\\
&\stackrel{[\partial_i,\partial_j]=0}{=} T(\H(K,\partial_i))+K(\H(\partial_i, T))+\partial_i \H(T,K).
\end{align*}
Since $\L_T \H = \L_K \H =0$, this shows $\partial_i (\H(T,K))=0$, i.e.,$\H(T,K)$ is a constant on $\M$.

On the other hand, on the horizon $\H^+$,  we know that $K$ is proportional to $L$.  Thanks to the
positive energy condition and  non-expansion condition of $\H^+$, we know that $\H (L,\cdot)$ vanishes on $\H^+$. This completes the proof.
\end{proof}
We now define the connected open space-time regions parametrized by $R \in \mathbb{R}^+$:
\begin{equation}
 \E_{R} =\{ p \in \E | y(p) < R\}.
\end{equation}
\begin{center}
\includegraphics[width=2.8in]{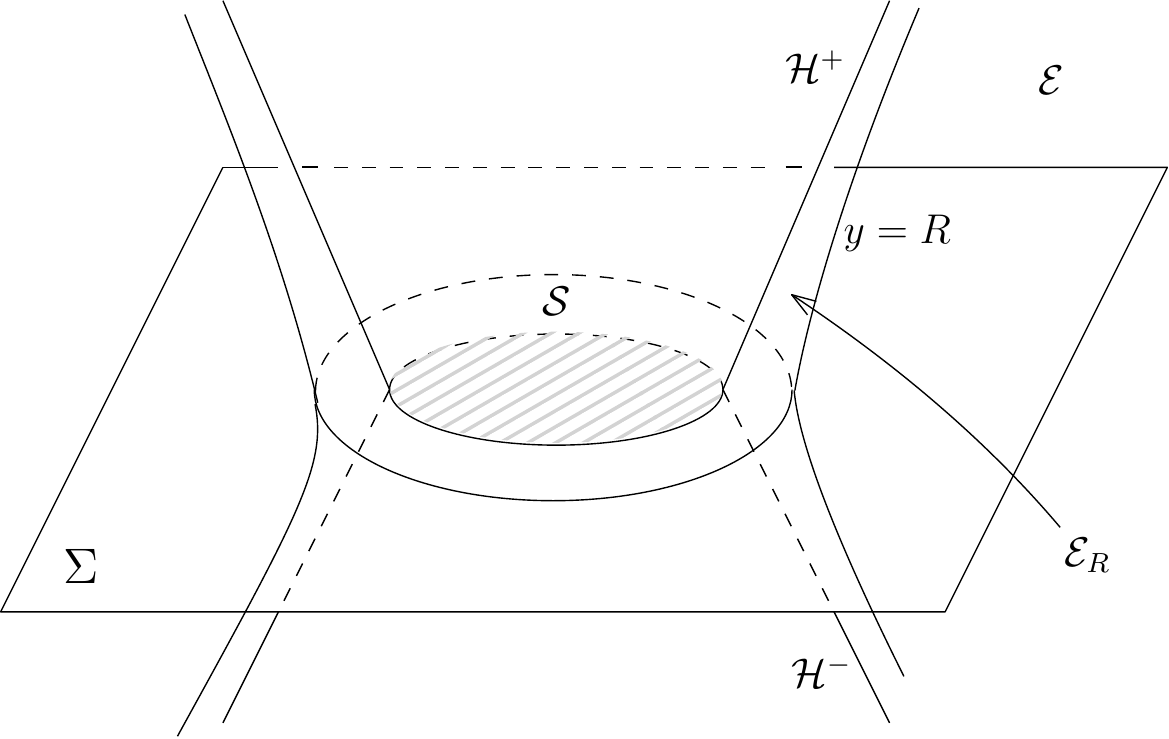}
\end{center}
Since $y$ is globally defined, we have $\displaystyle\bigcup_{R \geqslant y_+} \E_R = \E$. Near the bifurcate sphere $\S$,	 we also know that $T$ does not vanish in $\E_{y_+ + \varepsilon^2}$ provided $\varepsilon$ is sufficiently small. Therefore, we can solve $[T,K] = 0$ to construct a $K$ on $\E_{y_+ + \varepsilon^2}$ by using the initial data
in $\O$ (this is constructed in Proposition \ref{local}). It is easy to see that all the properties of $K$ which hold on $\O$ are perserved. This is the first step of induction for the construction of $K$. We now make the following claim which gives the construction of $K$ on the entire $\E$.
\begin{lemma}\label{ind}
 For all $R \geqslant y_+ + \varepsilon^2$, there is a smooth Killing vector field $K$ defined on $\E_R$ which agrees the one
defined on $\O$ in Proposition \ref{local}. Moreover, we have
 \begin{equation*}
 [T,K]=0, \qquad K(y)=0.
 \end{equation*}
\end{lemma}
We will use an induction agrument: the lemma holds for $R = y_+ + \varepsilon^2$. If we assume the lemma holds for some value $R_0$, we will show that it is still true for $R=R_0+\delta$ provided $\delta$ is suitably small. Therefore, the standard continuity argument will provide a proof of the lemma.

We first derive an ordinary differential equation (ODE) for $K$ in the region $\E_{R_0}$. This equation motivates the construction in $\E_{R_0+\delta}$. We compute $[K,\nabla y]$ as follows:
\begin{align*}
 [K,\nabla y]_b &=K^a \nabla_a \nabla_b y - \nabla^b \nabla_a K_b \stackrel{Killing}{=} K^a \nabla_a \nabla_b y +\nabla^b \nabla_b K_a\\
&=K^a \nabla_b \nabla_a y +\nabla^b \nabla_b K_a\\
&=\nabla_b (K^a \nabla_a y) =\nabla_b(K(y))=0
\end{align*}
Thus,  $[K, \nabla y] = 0$ in $\E_{R_0}$.

For a given small number $\deltab$, we define $\O_{\deltab, R_0} = \displaystyle\bigcup_{p \in \partial \E_{R_0} \cap \Sigma} B_{\deltab} (p)$, where $\Sigma$ is the given initial space-like surface and $B_{\deltab} (p)$ is a coordinate ball of radius $\bar{\delta}$. According to the induction hypothesis, the Hawking vector field $K$ has already been defined on $\E_{R_0} \cap \O_{\deltab, R_0}$. We shall extend it to the full open set $\O_{\deltab, R_0}$ by solving an ODE system.
\begin{lemma} If $\deltab$ is sufficiently small, then there is a smooth vector field $Y$ and a smooth
extension of $K$ to $\O_{\deltab, R_0}$, such that
\begin{equation}\label{requirements}
 \nabla_Y Y =0, \qquad [K,Y]=0, \qquad [T,Y]=0.
\end{equation}
\end{lemma}
\begin{proof}
If $\deltab$ is small, we can solve directly the geodesic equation $\nabla_Y Y =0$ with initial data $Y = \nabla y$ on $\partial \E_{R_0}$ to construct $Y$. So the first equation in \eqref{requirements} on $Y$ holds. To construct $K$, one first check that $[K,Y] = 0$ in $E_{R_0} \cap \O_{\deltab, R_0}$. It holds on $\partial \E_{R_0}$ since $K$ is tangent to $\E_{R_0}$ (this can be derived from the fact that $K(y) = 0$) and $Y = \nabla y$ on this surface.

On the other hand, let $W=[K,Y]$, then $W$ solves an ODE on $\E_{R_0} \cap \O_{\deltab, R_0}$ (where $K$ is Killing):
\begin{equation*}
 \nabla_Y (\L_K Y) = \L_K (\nabla_Y Y)-\nabla_{\L_K Y}Y=-\nabla_{\L_K Y}Y,
\end{equation*}
i.e.,
\[\nabla_Y W=-\nabla_{W}Y.\]
On the right hand side of the above expressin, it is tensorial in $W$, i.e., it does not involve the derivatives of $W$. Hence, this is an ODE for $W$. By virtue of the uniqueness of the solution of an ODE system, $W\equiv 0$.

Similar argument also shows that $[T,Y]=0$. Now we can solve the equation $[K,Y]=0$ to extend $K$ to the full neighborhood. This completes the proof.
\end{proof}

Although we have extended $K$ to the full open set $\O_{\deltab, R_0}$, we still need to show that this extension is a Killing vector field. The strategy is as follows: We consider the one parameter group $\phi_\tau$ generated by $K$. To show that $K$ is Killing, it suffices to show that for $\tau$ small enough, $\phi_\tau$ is an isometry, i.e. $\phi_\tau^* g =g$. Recall that we also want to show $K$ also preserves $H$, i.e., $\L_K H =0$, so in addition to $\phi_\tau^*
g =g$, we have to prove $\phi_\tau^* H =H$ at the same time. According to the covariant nature of the Einstein-Maxwell equations, we know that both $(g,H)$ and $(\phi_\tau^* g, \phi_\tau^* H)$ verify the Einstein-Maxwell equations. Moreover, $(g,H)$ and $(\phi_\tau^* g, \phi_\tau^* H)$ coincide on $\E_{R_0}$. To capture the stationary nature of the space-time, we
also define $T_\tau =(\phi_{-\tau})_* T$. Since $[K,Y]=0$, we have
\begin{equation*}
 (\phi_{-\tau})_* Y =Y.
\end{equation*}
Thus, we still have $\nabla^\tau_Y Y =0$ where $\nabla^\tau$ is the corresponding Levi-Civita connection of the metric $\phi_\tau^* g$. Lemma \ref{ind} on the open set $\O_{\deltab, R_0}$
will follow immediately (if one sets $g' = \phi_\tau^* g$, $H' =
\phi_\tau^* H$ and $T' = T_\tau$) from the following proposition (the proof is based on the $T$-conditional pseudo-convexity for $y$ and it is referred to the next section):
\begin{proposition}\label{tcontionaltheorem}
Assume $(g',H')$ is defined on $\O_{\deltab, R_0}$ solving the
Einstein-Maxwell equations and $T'$ is a Killing vector such that
$\L_{T'}H' =0$. If the following conditions hold:
\begin{align}
 \begin{cases}
  g = g' &\quad \text{and} \quad T'=T \quad \text{in} \quad \E_{R_0}, \\
  \nabla'_Y Y =0,&
\end{cases}
\end{align}
where $\nabla'$ is the Levi-Civita connection defined by $g'$. Then we have $g'=g$, $H'=H$ and $T'=T$ in an open set $\O_{\delta, R_0}\subset \O_{\deltab, R_0}$.
\end{proposition}
In view Lemma \ref{ind}, the part $[K,T] =0$ (the part $K(y)$=0 follows immediately from this by using the same argument as in the proof of Proposition \ref{local}) has already been proved, since we proved that $T_\tau =(\phi_{-\tau})_* T = T$. We then can extend $K$ to $\E_{R_0 + \deltab}$ by solving $[K,T]=0$.
{}

Now we have completed the proof of Lemma \ref{ind} and the construction of the Hawking vector field on
the domain of outer communication (Proposition \ref{tcontionaltheorem} will be proved in the next section).

\section{Proof of Proposition \ref{tcontionaltheorem} and the construction of rotational symmetry}
In the course of the proof of Proposition \ref{tcontionaltheorem}, for all the equations, the exact coefficients of the expressions are irrelevant. Only the structure of the equations is important. For this purpose, we shall use again the $*$ notation to simplify the expression without loss of information, e.g., we write $2A_{ab}B^{ab}$ as $A*B$ or $100A_{ab}B^{b}{}_{c}C_{b}{}^{c}$ as $A*B*C$.

For a given point $p_0$, we define a coordinate system $(x_k)$ where $k = 1,2,3,4$ on a neighborhood $\O(p_0)$ of $p_0$. We make the following important remark: the coordinates system is chosen for both metrics $g$ and $g'$. In this way, we can compare the various components of various geometric quantities for two different space-time $(\M,g,H)$ and $(\M,g',H')$ in this neighborhood. In this section, we will keep shrinking the neighborhoods around $p_0$; to simplify notations, we shall keep denoting such neighborhoods by $\O(p_0)$.

We define two vector fields $e_4 = e'_4 = Y$ and then fix two smooth frames $(e_1,\cdots,e_4)$ and $(e'_1,\cdots,e'_4)$ via parallel transport in a neighborhood $\O(p_0)$ of $p_0$ by using the same initial data on $\O(p_0) \cap E_{R_0}$ but with two different Levi-Civita connections:
\begin{align}\label{def_frame}
 \begin{cases}
  \nabla_Y e_a &= 0 \qquad \text{on} \quad \O(p_0),\\
  \nabla'_Y e'_a &= 0 \qquad \text{on} \quad \O(p_0), \qquad \qquad a=1,2,3,4\\
  e_a &= e'_a  \qquad \text{on} \quad \O(p_0) \cap E_{R_0}.
 \end{cases}
\end{align}
We now write all the geometric quantities in these two frames.

Let $g_{ab} =g(e_a,e_b)$ and $g'_{ab} =g'(e'_a,e'_b)$. We have
\begin{equation*}
 Y(g_{ab})= g(\nabla_Y e_a, e_b)+g(\nabla_Y e_b, e_a)=0
\end{equation*}
Similarly, we have $Y(g'_{ab})=0$. Since $g_{ab}$ and $g'_{ab}$ agree on $\O(p_0) \cap E_{R_0}$, therefore $g_{ab} \equiv g'_{ab}$. We denote $h_{ab} = g_{ab} = g'_{ab}$. It follows that
\begin{equation}\label{Yh}
 Y(h_{ab})= 0 \qquad \text{on} \quad \O(p_0).
\end{equation}
We now define the Christoffel symbols, curvature tensors, Maxwell 2-forms and their differences (up to one derivative) in the given frames:
\begin{align*}
\begin{cases}
\Gamma^{c}_{ab} &= g(\nabla_{e_a}e_b, e_c),\qquad \Gamma'^{c}_{ab} = g'(\nabla'_{e'_a}e'_b, e'_c), \\
\delta \Gamma^{c}_{ab} &= \Gamma'^{c}_{ab}-\Gamma^{c}_{ab}, \qquad (\partial \delta \Gamma)_{kab}^{c} = \partial_k (\Gamma'^{c}_{ab}-\Gamma^{c}_{ab}),\\
R_{abcd} &= g(R(e_a,e_b)e_c, e_d), \qquad R'_{abcd} = g'(R'(e'_a,e'_b)e'_c,e'_d), \\
\delta R_{abcd} &= R'_{abcd}-R_{abcd}, \qquad (\partial \delta R)_{kabcd} = \partial_k( R'_{abcd}-R_{abcd}),\\
H_{ab} &= H(e_a, e_b), \qquad H'_{ab} = H'(e'_a,e'_b), \\
\delta H_{ab} &= H'_{ab}-H_{ab},\qquad ( \partial \delta H)_{kab} = \partial_k(H'_{ab}-H_{ab}),\\
\nabla H_{abc} &= [\nabla_{e_c}H](e_a, e_b), \qquad \nabla' H'_{abc} = [\nabla_{e'_c}H'](e'_a,e'_b), \\
\delta \nabla H_{abc} &= \nabla H_{abc}-\nabla' H'_{abc},\qquad ( \partial \delta \nabla H)_{kabc} = \partial_k(\nabla' H'_{abc}-\nabla H_{abc}).
\end{cases}
\end{align*}
We remark that $\nabla H$ are treated as an independent geometric quantity although it can be derived from $H$. To compare the geometries, we also need to decompose the frames $\{e_a\}$ and $\{e'_a\}$ in terms of local coordinates system. In fact, we will prove not only the geometric quantities but also the frames are the equal at the same time. We then define the differences between frames, namely,
\begin{align*}
\begin{cases}
e_a &= e_a ^k \partial_k, \qquad e'_a = {e}_a '^k \partial_k, \\
(\delta e)_a^k &= {e}_a '^k \partial_k-e_a ^k \partial_k, \qquad (\partial \delta e)_{la}^k = \partial_l({e}_a '^k \partial_k-e_a ^k \partial_k).
\end{cases}
\end{align*}
We now derive the hyperbolic-ODE systems equations for vairous geometric quantities.

We first derive equations for curvatures and Maxwell 2-forms. We take the divergence of the second Bianchi identity:
\begin{equation*}
\nabla_a R_{bcde} + \nabla_b R_{cade}+\nabla_c R_{abde}=0.
\end{equation*}
By commuting derivatives,  we obtain the following schematic formula:
\begin{equation}\label{boxR}
 \square_g R_{abcd} = (R*R)_{abcd}+ \nabla_c \nabla_d T_{ab},
\end{equation}
where we replace the Ricci tensor $R_{ab}$ by the energy-momentum tensor according to the Einstein-Maxwell equations. One then computes the Hessian of $T$:
\begin{align*}
\nabla_c \nabla_d T_{ab} &= 2H_{b e} \nabla_c \nabla_d H_a{}^e  + 2H_{a}{}^{e} \nabla_c \nabla_d H_{be} + 2\nabla_d H_a{}^e \nabla_c H_{be} + 2\nabla_c H_a{}^e \nabla_d H_{be}\\
&\quad-g_{ab}(H^{ef}\nabla_c \nabla_d H_{ef} + \nabla_c H_{ef} \nabla_d H^{ef})\\
&=(H*\nabla^2 H)_{abcd}+ (\nabla H*\nabla H)_{abcd}
\end{align*}
Together with \eqref{boxR}, we conclude that
\begin{equation}\label{eqRR}
 \square_g R_{abcd} = (R*R)_{abcd}+ (H*\nabla^2 H)_{abcd}+ (\nabla H*\nabla H)_{abcd}.
\end{equation}
We take the divergence of $\nabla_{[a} H_{bc]} =0$ (this is one of the Maxwell equations) to derive:
\begin{align*}
 \nabla^a \nabla_a H_{bc} &= -\nabla^a \nabla_b H_{ca} - \nabla^a \nabla_c H_{ab}\\
 &= R^{a}{}_{bcd} H^{d}{}_{a}+ R^{a}{}_{bad} H_{c}{}^{d}.
\end{align*}
For the last step, we used the fact the $\nabla^a H_{ab} =0$. Therefore, schematically, we have
\begin{equation}\label{eqFF}
\square_g H_{ab} = (R * H)_{ab}.
\end{equation}
By commuting $\nabla$, we can similarly derive
\begin{equation}\label{eqDDF}
 \square_g (\nabla H)_{ab} = (R * \nabla H)_{ab} + (\nabla R * H)_{ab}
\end{equation}
We summarize \eqref{eqRR}, \eqref{eqFF} and \eqref{eqDDF} in following system of equations
\begin{align}\label{eqmainn}
\begin{cases}
    \square_g R_{abcd} &= (R*R)_{abcd}+ (H*\nabla^2 H)_{abcd}+ (\nabla H* \nabla H)_{abcd},\\
        \square_g \H_{ab} &= (R * H)_{ab},\\
    \square_g (\nabla H)_{ab} &= (R * \nabla H)_{ab} + (\nabla R *H)_{ab}.
        \end{cases}
\end{align}
For the space-time $(g',H')$, we also have
\begin{align}\label{eqmainn'}
\begin{cases}
    \square_{g'} R'_{abcd} &= (R'*R')_{abcd}+ (H'*\nabla'^2 H')_{abcd}+ (\nabla' H'* \nabla' H')_{abcd},\\
        \square_{g'} \H'_{ab} &= (R' * H')_{ab},\\
    \square_{g'} (\nabla' H')_{ab} &= (R' * \nabla' H')_{ab} + (\nabla' R' *H')_{ab}.
        \end{cases}
\end{align}
We now take the difference of these two system of equations. To
illustrate the idea, we compute the the difference between
$\square_g R_{abcd}$ and $\square_{g'} R'_{abcd}$ in detail:
\begin{align}\label{BBoxR}
 \square_g R_{abcd} - \square_{g'} R'_{abcd} &= [(R*R)_{abcd}-(R'*R')_{abcd}] + [(H*\nabla^2 H)_{abcd}-(H'*\nabla'^2 H')_{abcd}]\notag\\
&\quad + [(\nabla H* \nabla H)_{abcd}-(\nabla' H'* \nabla' H')_{abcd}].
\end{align}
First of all, we study the term $R*R-R'*R'$ and we have
\begin{align*}
 (R*R-R'*R')_{abcd} &= [R*(R-R')-(R-R')*R']_{abcd}\\
&= (R *\delta R +  \delta R * R')_{abcd}.
\end{align*}
The metrics $g$ and $g'$ are given objects, so their corresponding curvatures are known, one can estimate $R_{abcd}$ and $R'_{abcd}$ in $L^{\infty}$-norm, so we have
\begin{align*}
 |(R*R-R'*R')_{abcd}| \lesssim |\delta R|,
\end{align*}
where the inequality is understood in the point-wise sense.

We turn to the term $H*\nabla^2 H-H'*\nabla'^2 H'$. To simplify the expressions, we omit the indices for tensor components:
\begin{align*}
&\ \ \ \ H*\nabla^2 H-H'*\nabla'^2 H'\\
&=(H-H')*\nabla^2 H +  H'*[(\nabla -\nabla')\nabla H] +H'*\nabla'(\nabla H- \nabla 'H')
\end{align*}
For the first term on the right hand side of the equation, we bound $\nabla^2 H$ in $L^{\infty}$-norm. Therefore, this term is bounded by $|\delta H|$ up to a constant; for the second one term, we recall that $\nabla -\nabla'$ is a 0-order differential operator bounded by $|\delta \Gamma|$ and one can still bound $H'$ and $\nabla H$ in
$L^{\infty}$-norm, so we can bound the second term by $|\delta \Gamma|$; for the third one, notice that for a given tensor $A$, $\nabla'$ differs from the standard coordinate derivative by a collections of Christoffel symbols, so we have
\begin{equation*}
 |\nabla' A| \lesssim |A| + |\partial A|.
\end{equation*}
It implies the third term is bounded by $|\delta \nabla H| + |\partial \delta \nabla H|$. Finally, we have
\begin{align*}
 |H*\nabla^2 H-H'*\nabla'^2 H'| \lesssim |\delta \Gamma|+|\delta H|+|\delta \nabla H| + |\partial \delta \nabla H|.
\end{align*}
Similarly, we can derive
\begin{equation*}
 |\nabla H* \nabla H-\nabla' H'* \nabla' H'| \lesssim |\delta \nabla H|.
\end{equation*}
We turn to $\square_g R - \square_{g'} R'$:
\begin{align*}
 \square_{g'} R' - \square_{g} R &=  (\square_{g'} - \square_{g}) R' +  \square_{g}(R'-R)\\
 &=(\square_{g'} - \square_{g}) R' +  \square_{g}\delta R.
\end{align*}
Recall that $\square_{g'} - \square_{g}$ is a first order
differential operator. It is easy to see that $(\square_{g'} -
\square_{g}) R'$ ($R'$ is estimated in $L^\infty$-norm) is bounded
by $|\delta \Gamma| + |\partial \delta \Gamma|$. Finally, we
conclude from \eqref{BBoxR} that
\begin{equation}
  \square_{g}\delta R \lesssim  |\delta \Gamma| + |\partial \delta \Gamma| + |\delta H|+|\delta \nabla H| + |\partial \delta \nabla H| + |\delta  R|.
\end{equation}
Now we take the difference of \eqref{eqmainn} and \eqref{eqmainn'},
similar computations give the following set of differential
inequalities:
\begin{align}\label{MainPDE}
\begin{cases}
    \square_g (\delta R)_{abcd}  &\lesssim |\delta \Gamma| + |\partial \delta \Gamma| + |\delta H|+|\delta \nabla H| + |\partial \delta \nabla H| + |\delta  R|,\\
    \square_g  \delta H_{ab} &\lesssim |\delta H|+ |\delta  R|,\\
\square_g (\delta \nabla H)_{ab} &\lesssim |\delta \Gamma|+ |\delta H|+|\delta \nabla H|+|\delta R| +|\partial \delta R|.
        \end{cases}
\end{align}

Now we derive a system of ordinary differential inequalities along the vector field $Y$ for the differences between Christoffel symbols $\delta \Gamma$, $ \partial \delta \Gamma$ and the differences between two given frames $\delta e$, $\partial \delta e$. We omit the indices when there is no confusion, as an example, we use $\delta e$ to denote the collection of $(\delta e)_a^k$'s for all the $a$ and $k$. We first compute
$Y(\Gamma_{ab}^{c})$:
\begin{align*}
Y(\Gamma^{c}_{ab}) &= Y(g(\nabla_{e_a}e_b, e_c)) = g(\nabla_{e_4} \nabla_{e_a}e_b, e_c)+g(\nabla_{e_a}e_b, \nabla_{e_4} e_c)\\
              &\stackrel{\nabla_{Y}e_a =0}{=}R_{4abc} + g(\nabla_{[e_4,e_a]} e_b, e_c)+g(\nabla_{e_a}e_b, \nabla_{e_4} e_c)\\
              &= R_{4abc} + \Gamma^{d}_{4a}\Gamma^{c}_{db}-\Gamma^{d}_{a 4}\Gamma^{c}_{db} + g_{d f}\Gamma^{f}_{ab} \Gamma^{d}_{4c}.
\end{align*}
Recall that $g_{ab}$ is the same for both metrics, so schematically, we have
\begin{equation*}
Y(\Gamma^{c}_{ab}) = R_{4abc} + (\Gamma * \Gamma)_{ab}^c.
\end{equation*}
Similarly, we have
\begin{equation*}
Y(\Gamma'^{c}_{ab}) = R'_{4abc} + (\Gamma' * \Gamma')_{ab}^c.
\end{equation*}
One can take the difference to get
\begin{equation}\label{YGamma}
 Y(\delta \Gamma_{ab}^c) \lesssim |\delta R|+|\delta \Gamma|.
\end{equation}
One can also first take a derivative in $\partial_k$ direction and then take the difference, this procedure gives
\begin{equation}\label{YpartialGamma}
 Y( (\partial\delta \Gamma)_{kab}^c) \lesssim |\delta R|+|\partial \delta R|+ |\delta \Gamma|+|\partial \delta \Gamma|.
\end{equation}
To compute $Y(e_{a}^k)$, we consider $[Y,e_a]$:
\[[Y,e_a] = -\nabla_Y e_a =- \Gamma_{4a}^b e_b = - \Gamma_{4a}^b e^k_b \partial_k.\]
This implies
\begin{equation*}
 Y^j \partial_j (e_a^k)-e_a^j \partial_j(Y^k)= - \Gamma_{4a}^b e^k_b \partial_k,
\end{equation*}
i.e.,
\begin{equation*}
 Y(e_a^k)=\partial_j(Y^k) e_a^j - \Gamma_{4a}^b e^k_b .
\end{equation*}
Schematically, we have
\begin{equation*}
 Y(e_a^k)= e + \Gamma * e.
\end{equation*}
We also have $ Y({e'}_a^k)= e' + \Gamma * e'$, one can then take one more $\partial_k$ derivative and take difference as before, together with \eqref{YGamma} and \eqref{YpartialGamma}, we have the following set of ordinary differential inequalities:
\begin{align}\label{MainODE}
\begin{cases}
     Y(\delta \Gamma_{ab}^c) &\lesssim |\delta R|+|\delta \Gamma|,\\
     Y((\partial\delta \Gamma)_{kab}^c) & \lesssim |\delta R|+|\partial \delta R|+ |\delta \Gamma|+|\partial \delta \Gamma|,\\
     Y(\delta e_{a}^k) &\lesssim |\delta e|+|\delta \Gamma|,\\
     Y((\partial \delta) e_{la}^k) &\lesssim |\delta e|+|\delta \Gamma|+|\partial \delta e|+|\partial \delta \Gamma|.
\end{cases}
\end{align}

Finally, we shall derive an additional system of ordinary differential inequalities to include the information the Killing vector fields $T$ and $T'$. We now define the differences for various quantities related to these symmetries: \begin{align*}
\begin{cases}
T_a &= g(T,e_a),\quad T'_a = g(T',e'_a), \quad \delta T_a = T'_a -T_a;\\
F_{ab} &=g(\nabla_{e_a} T, e_b), \quad F'_{ab} =g(\nabla_{e'_a} T', e'_b), \quad \delta F_{ab} = F'_{ab}-F_{ab}.
\end{cases}
\end{align*}
We first compute $Y(T_a)$:
\begin{equation*}
 Y(T_a) = Y(g(T,e_a)) \stackrel{\nabla_Y e_a =0}{=} g(\nabla_Y T, e_a)= F_{4a}.
\end{equation*}
Similar computation gives $Y(T'_a) = F'_{4a}$, by taking the difference, we have
\begin{equation}\label{Yt}
 Y(\delta T_a) \lesssim |\delta F|.
\end{equation}
Now we compute $Y(F_{ab})$:
\begin{align*}
 Y(F_{ab})&\stackrel{\nabla_Y e_a =0}{=} \nabla_Y F_{ab} \stackrel{Killing}{=}R_{ab4c} T^c \stackrel{g_{ab} = h_{ab}}{=} h^{cd}R_{ab4c} T_d
\end{align*}
Similarly, we have $Y(F'_{ab})=h^{cd}R'_{ab4c} T'_d$, by taking the difference, we have
\begin{equation}\label{YF}
 Y(\delta F_{ab}) \lesssim |\delta T|+ |\delta R|.
\end{equation}
Finally, we use the following identities:
\begin{align*}
 0 & \stackrel{Killing}{=} \L_T R_{abcd} \\
 &\stackrel{\nabla_Y e_a =0}{=} \nabla_T R_{abcd} + \nabla_a T^f R_{fbcd}  + \nabla_b T^f R_{afcd} +\nabla_c T^f R_{abfd} +\nabla_d T^f R_{abcf}.
\end{align*}
Similarly, we have
\begin{align*}
 0 = \nabla_{T'} R'_{abcd} + \nabla_a {T'}^f R'_{fbcd}  + \nabla_b {T'}^f R_{afcd} +\nabla_c {T'}^f R_{abfd} +\nabla_d {T'}^f R_{abcf}.
\end{align*}
After taking the difference, we have
\begin{align*}
 |\nabla_T R_{abcd}- \nabla_{T'} R'_{abcd}| \lesssim |\delta F| + |\delta R|.
\end{align*}
This gives
\begin{align*}
 T(\delta R) \lesssim |\delta \Gamma| + |\delta T| +  |\delta F| + |\delta e|+ |\delta R|.
\end{align*}
Similarly, we have
\begin{align*}
 T(\delta H) \lesssim |\delta \Gamma| + |\delta T| +  |\delta F| + |\delta e|+ |\delta H|,
\end{align*}
and
\begin{align*}
 T(\delta \nabla H) \lesssim |\delta \Gamma| + |\delta T| +  |\delta F| + |\delta e|+ |\delta H|+|\delta \nabla H|+|\delta R|.
\end{align*}
Together with \eqref{Yt} and \eqref{YF}, we have the following system of differential inequalities:
\begin{align}\label{MainKillingODE}
 \begin{cases}
   Y(\delta T_a) &\lesssim |\delta F|,\\
   Y(\delta F_{ab}) &\lesssim |\delta T|+ |\delta R|,\\
   T(\delta R) & \lesssim |\delta \Gamma| + |\delta T| +  |\delta F| + |\delta e|+ |\delta R|,\\
   T(\delta H) &\lesssim |\delta \Gamma| + |\delta T| +  |\delta F| + |\delta e|+ |\delta H|,\\
   T(\delta \nabla H) &\lesssim |\delta \Gamma| + |\delta T| +  |\delta F| + |\delta e|+ |\delta H|+|\delta \nabla H|+|\delta R|.
 \end{cases}
\end{align}

\subsection{Proof of Proposition }
We recall the following unique continuation lemma due to Alexakis, Ionescu and Klainerman:
\begin{lemma}[Lemma 5.5 in \cite{Alexakis_Ionescu_Klainerman_Perturbation}]
Assume $\delta>0$, $p_0\in\delta_{\Sigma_1}(\mathcal{U}_{R_0})$ and $G_i,H_j:B_{\delta}(p_0)\to\mathbb{R}$ are smooth functions, $i=1,\ldots,I$, $j=1,\ldots,J$. Let $G=(G_1,\ldots,G_I)$, $H=(H_1,\ldots,H_J)$, $\partial G=(\partial_0G_1,\ldots,\partial_4G_I)$ and assume that, in $B_{\delta}(p_0)$,
\begin{equation*}
\begin{cases}
&\square_g G=\mathcal{M}_\infty(G)+\mathcal{M}_\infty(\partial G)+\mathcal{M}_\infty(H);\\
&\T(G)=\mathcal{M}_\infty(G)+\mathcal{M}_\infty(H);\\
&\overline{Y}(H)=\mathcal{M}_\infty(G)+\mathcal{M}_\infty(\partial G)+\mathcal{M}_\infty(H).
\end{cases}
\end{equation*}
Assume that $G=0$ and $H=0$ in $B_{\delta}(p_0)\cap\E_{R_0}=\{x\in B_{\delta}(p_0):y(x)<R_0\}$. Then $G=0$ and $H=0$ in $B_{\widetilde{\delta}}(p_0)$ for some $\widetilde{\delta}\in(0,\delta)$ sufficiently small.
\end{lemma}

The notation $\mathcal{M}_\infty(G)$ means $\lesssim |G|$. Combining \eqref{MainPDE}, \eqref{MainODE} and \eqref{MainKillingODE}, in view of the the $T$-conditional pseudo-convexity condition for $y$, we can apply the lemma to conclude that $H_{ab}= H_{ab}'$ and $R_{abcd} = R'_{abcd}$.

This completes the proof of Proposition \ref{tcontionaltheorem}.

\subsection{The rotational symmetry}
In last subsection of the work, we construct the rotational vector field $Z$ on the entire domain of outer communication $\E$. As we described before, this additional symmetry allows one to use theresult of Bunting to conclude that the space-time $(\M, g, H)$ is in fact isometric to a Kerr-Newman solution.

Recall the local construction in Proposition \ref{local}, one can find $\lambda_0 \in \mathbb{R}$ such that around the bifurcate sphere $\S$, the vector field $Z = T + \lambda_0 K$ is a Killing vector field with closed orbits and with a period $\tau_0$. Now we construct $Z$ on the entire exterior region by the same formula (both $T$ and $K$ are globally defined):
\begin{equation*}
 Z = T + \lambda_0 K
\end{equation*}
Since both $K$ and $T$ are Killing, we know that $Z$ is also Killing. In addition, the following identities hold immediately since they hold for both $K$ and $T$:
\begin{equation}
[T, Z] = [K, Z]=[Z, \nabla y]=0, \qquad Z(y) = 0.
\end{equation}
Let $\phi_\tau$ be the one parameter isometry group generated by $Z$. From the local property , we know that $\phi_{\tau_0} = {\rm Id}$ around the bifurcate sphere $\S$. Since $Z$ commutes with $T$, we can argue as before to show that  $\phi_{\tau_0} = {\rm Id}$ on $\E_{y_+ +\varepsilon^2}$. We now use again the recursive argument to show that $\phi_{\tau_0} = Id$ for on each $\E_R$ where $R \geqslant y_+ + \varepsilon^2$. Assume for a given $R_0$ this conclusion holds, namely,
\begin{equation*}
 \phi_{\tau_0} (p) = p \qquad \text{for any} \quad p\in \E_{R_0},
\end{equation*}
Let  $\psi_\tau$ to be one parameter group of diffeomorphisms generated by $\nabla y$. Since $[Z, \nabla y] =0$, we know that for
small $\tau > 0$
\begin{equation}
  \phi_{\tau_0} \circ \psi_\tau = \psi_\tau \circ \phi_{\tau_0}
\end{equation}
Apply this identity to a point $p\in \E_{R_0}$, one has
\begin{equation}
 \phi_{\tau_0}(\psi_\tau(p)) = \psi_\tau(\phi_{\tau_0}(p))= \psi_\tau(p).
\end{equation}
Since $y$ has no critical point, the flow $\psi_\tau$ moves $p$ from $\E_{R_0}$ towards the outside of $\E_{R_0}$, i.e. $\phi_{\tau_0}
= {\rm Id}$ on a larger neighborhood of $\E_{R_0}$. In this way, we can show that there is a $\delta$, such that $\phi_{\tau_0} = {\rm Id}$ for on each $\E_{R_0+ \delta}$. A simple continuity argument shows $\phi_{\tau_0} = {\rm Id}$ on the entire domain of outer communication.

\end{document}